\documentclass[12pt]{article}%
\usepackage[nosort]{cite}
\usepackage{graphicx}
\usepackage{multicol}
\usepackage{amsfonts}
\usepackage{amssymb}
\usepackage{amsmath}
\usepackage{amsthm}
\usepackage{dsfont}
\usepackage{heck}
\usepackage{afterpage}
\usepackage{setspace}
\usepackage{verbatim}
\usepackage{color}
\usepackage{longtable}
\usepackage{float}
\usepackage{subcaption}
\usepackage{epsfig}
\usepackage{epstopdf}
\usepackage{adjustbox}
\usepackage{todonotes}
\usepackage{tikz}
\usepackage[margin=1in]{geometry}
\usepackage{titletoc}
\usepackage{hyperref}
\usepackage{mathdots}
\usepackage{pgfplots}
\usepgfplotslibrary{colormaps,groupplots}
\usepackage{tikz}
\usetikzlibrary{arrows,calc,matrix,decorations.markings}
\tikzstyle shortorb=[minimum size=2mm,inner sep=0pt,outer sep=0pt,shape=circle,fill=black]
\tikzstyle longorb=[minimum size=2mm,inner sep=0pt,outer sep=0pt,shape=circle,fill=white,draw]
\tikzstyle dot=[minimum size=1mm,inner sep=0pt,outer sep=0pt,shape=circle,fill=black]
\tikzstyle cross=[minimum size=1mm,inner sep=0pt,outer sep=0pt,shape=rectangle,fill=black]
\tikzstyle dynnode=[circle, inner sep=0pt, dashed, text width=1.6cm, align=center, draw=black, thick, fill=white]
\tikzset{>=stealth}
\tikzset{->-/.style={decoration={
  markings,
  mark=at position #1 with {\arrow[scale=1.8]{>}}},postaction={decorate}}}
\tikzset{
  lineC/.style={
	  color of colormap={#1},
		draw=.!80!black,
		fill=.!80!white,
},}%
\usepackage{xstring}
\tikzset{pmseq/.style={decoration={
  markings,
  mark=between positions 1em and 1 step 1.2em with {
    \node [anchor=center,draw,circle,minimum size=0mm,inner sep=0pt,outer sep=0pt,scale=1.2,fill=white]  {\tt\StrChar{#1}{\pgfkeysvalueof{%
    /pgf/decoration/mark info/sequence number}}};}
  }, postaction={decorate}}
}
\providecommand{\U}[1]{\protect\rule{.1in}{.1in}}

\pdfoutput=1
\newsavebox{\mysavebox}

\usetikzlibrary[shapes.geometric]
\usetikzlibrary{positioning}
\usetikzlibrary{calc,intersections,through,backgrounds}
\tikzset{>=stealth}
\usetikzlibrary{decorations.pathmorphing}
\usetikzlibrary{decorations.markings}
\tikzset{>=stealth}
\hypersetup{colorlinks,citecolor=black,filecolor=black,linkcolor=black,urlcolor=black,pdftex}
\usetikzlibrary{decorations.markings}

\numberwithin{equation}{section}

\newcommand{\ba}{\begin{eqnarray}}
\newcommand{\ea}{\end{eqnarray}}

\newcommand{\be}{\begin{equation}}
\newcommand{\ee}{\end{equation}}

\tikzstyle{startstop} = [rectangle, rounded corners, minimum width=3cm, minimum height=1cm,text centered, draw=black, fill=blue!10]
\tikzstyle{startstop} = [rectangle, rounded corners, minimum width=3cm, minimum height=1cm,text centered, draw=black, fill=blue!10]
\tikzstyle{io} = [trapezium, trapezium left angle=70, trapezium right angle=110, minimum width=3cm, minimum height=1cm, text centered, draw=black, fill=blue!30]
\tikzstyle{process} = [rectangle, minimum width=3cm, minimum height=1cm, text centered, draw=black, fill=orange!30]
\tikzstyle{decision} = [diamond, minimum width=3cm, minimum height=1cm, text centered, draw=black, fill=green!30]
\tikzstyle{arrow} = [thick,->,>=stealth]

\newtheorem{lemma}{Lemma}
\newtheorem{theorem}{Theorem}
\newtheorem{corollary}{Corollary}
\newtheorem{definition}{Definition}

\begin{document}

\date{November 2017}

\title{Punctures and Dynamical Systems}

\institution{PENN}{\centerline{${}^{1}$Department of Physics and Astronomy, University of Pennsylvania, Philadelphia, PA 19104, USA}}

\institution{UNC}{\centerline{${}^{2}$Department of Physics, University of North Carolina, Chapel Hill, NC 27599, USA}}

\authors{Falk Hassler\worksat{\PENN , \UNC}\footnote{e-mail: {\tt fhassler@unc.edu}}
and Jonathan J. Heckman\worksat{\PENN}\footnote{e-mail: {\tt jheckman@sas.upenn.edu}}}

\abstract{With the aim of better understanding the class of 4D theories
generated by compactifications of 6D superconformal field theories (SCFTs),
we study the structure of $\mathcal{N}=1$ supersymmetric punctures for
class $\mathcal{S}_{\Gamma }$ theories, namely the 6D\ SCFTs obtained from
M5-branes probing an ADE\ singularity. For M5-branes probing a $\mathbb{C}^2 / \mathbb{Z}_{k}$
singularity, the punctures are governed by a dynamical system in which
evolution in time corresponds to motion to a neighboring node in an affine A-type quiver.
Classification of punctures reduces to
determining consistent initial conditions which produce periodic orbits. The
study of this system is particularly tractable in the case of a single
M5-brane. Even in this ``simple'' case, the solutions exhibit a remarkable
level of complexity: Only specific rational values for the initial momenta
lead to periodic orbits, and small perturbations in these values lead to
vastly different late time behavior. Another difference from half BPS punctures
of class $\mathcal{S}$ theories includes the appearance of a continuous complex
``zero mode'' modulus in some puncture solutions. The construction of punctures with
higher order poles involves a related set of recursion relations. The resulting structures
also generalize to systems with multiple M5-branes as well as probes of
D- and E-type orbifold singularities.}

\maketitle

\tableofcontents

\enlargethispage{\baselineskip}

\setcounter{tocdepth}{2}

\newpage

\section{Introduction}

Compactifications of higher-dimensional theories produce a wealth of insights
into the construction and study of lower-dimensional quantum field theories.
In some sense, the natural starting point for addressing many issues of
compactification is to start with the highest dimension quantum field theories
for which supersymmetry and conformal symmetry can be combined, namely
compactifications of six-dimensional superconformal field theories
(6D\ SCFTs). The first evidence for the existence of 6D\ SCFTs appeared in references
\cite{Witten:1995ex, Witten:1995zh, Strominger:1995ac, Seiberg:1996qx} (see also \cite{WittenSmall,
Ganor:1996mu,MorrisonVafaII,Seiberg:1996vs, Bershadsky:1996nu,
Brunner:1997gf, Blum:1997fw, Aspinwall:1997ye, Intriligator:1997dh, Hanany:1997gh}, and there has
recently been renewed interest in the subject, both in terms of classifying
the resulting theories \cite{Heckman:2013pva, Gaiotto:2014lca, DelZotto:2014hpa,
DelZotto:2014fia, Heckman:2015bfa, Bhardwaj:2015xxa, Chang:2017xmr}, as well
as extracting non-trivial data from these theories and their compactifications
\cite{Apruzzi:2013yva, Heckman:2014qba, DelZotto:2015isa, Gaiotto:2015usa, Ohmori:2015pua, Franco:2015jna, DelZotto:2015rca, Heckman:2015ola, Cordova:2015fha, Hanany:2015pfa, Aganagic:2015cta, Louis:2015mka, Ohmori:2015pia, Coman:2015bqq, Cremonesi:2015bld, Heckman:2016ssk, Cordova:2016xhm, Morrison:2016nrt,  Heckman:2016xdl, Cordova:2016emh, Kim:2016foj, Shimizu:2016lbw, Mekareeya:2016yal, DelZotto:2016pvm, Apruzzi:2016nfr, Razamat:2016dpl, Bah:2017gph, Mitev:2017jqj, Bah:2017wxp, DelZotto:2017pti, Apruzzi:2017iqe, Heckman:2017uxe, Kim:2017toz, Razamat:2017hda}.

For 4D theories obtained from 6D\ SCFTs, the defining data for a
compactification includes specifying a Riemann surface with punctures, namely
marked points with prescribed boundary conditions for various operators of the
6D\ SCFT. A particularly clean class of examples are 4D $\mathcal{N}=2$
theories obtained from compactifications of $\mathcal{N}=(2,0)$ theories on
Riemann surfaces with first order poles for operators at marked points. For a
class $\mathcal{S}$ theory associated with a Lie algebra $\mathfrak{g}_{ADE}$ of ADE type,
the punctures are then specified by
embeddings of the Lie algebra $\mathfrak{su}(2)$ in $\mathfrak{g}_{ADE}$.
These are in turn characterized by nilpotent orbits in the Lie algebra $\mathfrak{g}_{ADE}$,
and for A-type theories this has a simple pictorial
representation in terms of Young diagrams, as used for example in
reference \cite{Gaiotto:2009we}.

The vast number of additional 6D\ SCFTs with $\mathcal{N}=(1,0)$ supersymmetry
suggests a corresponding proliferation of possible 4D $\mathcal{N}=1$ theories
obtained from subsequent compactification. While it is still an open question
to determine the structure of punctures in all 6D SCFTs, in the special case
of the class $\mathcal{S}_{\Gamma}$ theories, namely M5-branes probing an
ADE\ singularity $\mathbb{C}^{2}/\Gamma_{ADE}$, the defining equations for
punctures have at least been worked out \cite{Gaiotto:2015usa, Heckman:2016xdl}.
What remains to be done, however, is develop a classification scheme for possible boundary
conditions, analogous to what exists for $1/2$ BPS\ simple punctures of class
$\mathcal{S}$ theories.

To a certain extent, all of the $1/2$ BPS\ punctures of class $\mathcal{S}%
_{\Gamma}$ theories descend from the special case of $\mathcal{N}=1$
supersymmetric punctures of A-type class $\mathcal{S}$ theories, subject to
the additional conditions imposed by a Douglas-Moore type orbifold projection
\cite{Douglas:1996sw}. In the special case of simple punctures, the $1/4$ BPS\ puncture
equations of class $\mathcal{S}$ theories are \cite{Xie:2013gma}
(see also \cite{Gaiotto:2015usa, Heckman:2016xdl}):
\begin{equation}
\lbrack\Sigma,Q]=Q\text{, \ \ }[\Sigma,\widetilde{Q}]=\widetilde{Q}\text{,
\ \ }[Q,\widetilde{Q}]=0\text{, \ \ }[Q,Q^{\dag}]+[\widetilde{Q}%
,\widetilde{Q}^{\dag}] = \Sigma, \label{ClassSpunctures}%
\end{equation}
where $\Sigma,Q$, and $\widetilde{Q}$ are $N\times N$ matrices with complex
entries, and $\Sigma$ is Hermitian. Making the replacement
$N\mapsto \left\vert \Gamma\right\vert N$ and also demanding these matrices
transform in suitable representations of $\Gamma$ then leads to $\mathcal{N}=1$
supersymmetric punctures of class $\mathcal{S}_{\Gamma}$ theories \cite{Gaiotto:2015usa, Heckman:2016xdl}.
In the special case where $[Q,\widetilde{Q}^{\dag}]=0$, the problem reduces to the
study of embeddings of $\mathfrak{su}(2)^{l}\rightarrow\mathfrak{su}(\left\vert \Gamma\right\vert N )$,
subject to the additional constraints of an orbifold projection.
This of course prompts the question as to whether such solutions are
\textquotedblleft typical.\textquotedblright

Our aim in this paper will be to take some preliminary steps in classifying
punctures for class $\mathcal{S}_{\Gamma}$ theories, focusing in particular on
the case of simple punctures for M5-branes probing $\mathbb{C}^{2} / \mathbb{Z}_k$.
Some of the solutions we present also embed in the D-and E-type
class $\mathcal{S}_{\Gamma}$ theories.

The main idea we introduce is that the puncture equations can be visualized in
terms of a discrete dynamical system. Indeed, instead of working in terms of
large $Nk\times Nk$ matrices, we can alternatively use a \textquotedblleft
quiver basis\textquotedblright\ in which the orbifold projection has already
been imposed. In this case, depicted in figure~\ref{fig:atype}, we can label
the nodes of the quiver by an index $i=1,...,k$, and visualize this index as
a time step in a discrete dynamical system. The evolution equations are, in
terms of matrices $q(i)$, $\widetilde{q}(i)$ and $p(i)$ respectively
associated with the parent matrices
$Q$, $\widetilde{Q}$ and $\Sigma$:%
\begin{align}
q(i) &= p(i)q(i)-q(i)p(i+1)  \\
-\widetilde{q}(i) &= \widetilde{q}(i)p(i)-p(i+1)\widetilde{q}(i)  \\
0 &= q(i)\widetilde{q}(i)-\widetilde{q}(i-1)q(i-1)  \\
p(i) &= \left[  q(i)q^{\dag}(i)-\widetilde{q}^{\dag}(i)\widetilde{q}(i)\right]
-\left[  q^{\dag}(i-1)q(i-1)-\widetilde{q}(i-1)\widetilde{q}^{\dag
}(i-1)\right] .
\end{align}
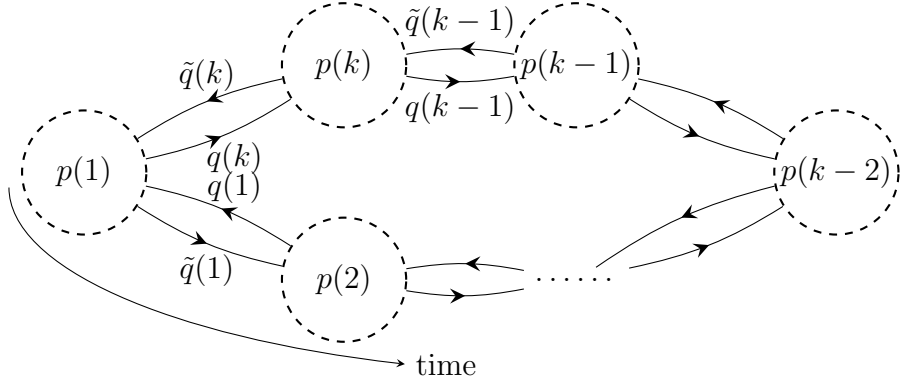
\begin{figure}
  \centering
  \begin{tikzpicture}
    \path (0,0) arc [start angle=0, end angle=360, x radius=5cm, y radius=1.5cm]
    node [pos=.0,dynnode] (nodekm2) {\mbox{$p(k-2)$}}
    node [pos=.2,dynnode] (nodekm1) {\mbox{$p(k-1)$}}
    node [pos=.3,dynnode] (nodek) {$p(k)$}
    node [pos=.500000,dynnode] (node1) {$p(1)$}
    node [pos=.7,dynnode] (node2) {$p(2)$}
    node [pos=.8] (dots) {$\dots \dots$};
    \draw[->-=.5] (nodekm1) to [bend right = 10] node[above] {$\tilde q(k-1)$} (nodek);
    \draw[->-=.5] (nodek)   to [bend right = 10] node[below] {$q(k-1)$} (nodekm1);
    \draw[->-=.5] (nodekm2) to [bend right = 10] (nodekm1);
    \draw[->-=.5] (nodekm1) to [bend right = 10] (nodekm2);
    \draw[->-=.5] (dots)    to [bend right = 10] (nodekm2);
    \draw[->-=.5] (nodekm2) to [bend right = 10] (dots);
    \draw[->-=.5] (node2)   to [bend right = 10] (dots);
    \draw[->-=.5] (dots)    to [bend right = 10] (node2);
    \draw[->-=.5] (node1)   to [bend right = 10] node[below] {$\tilde q(1)$} (node2);
    \draw[->-=.5] (node2)   to [bend right = 10] node[above,xshift=5pt,yshift=-2pt] {$q(1)$} (node1);
    \draw[->-=.5] (nodek)   to [bend right = 10] node[above] {$\tilde q(k)$} (node1);
    \draw[->-=.5] (node1)   to [bend right = 10] node[below,xshift=5pt,yshift=2pt] {$q(k)$} (nodek);
    \path[draw,->] (-11,-.2) arc [start angle=180, end angle=250, x radius=8cm, y radius=2.5cm]
    node [pos=1,anchor=west] {time};
  \end{tikzpicture}
  \caption{A-type quiver with the direction along the quiver interpreted as the time of the
  dynamical system.}\label{fig:atype}
\end{figure}
So, once we know the values of the matrices at timestep $i$, the evolution to
timestep $i+1$ is implicitly specified. A similar iteration procedure also holds
for punctures with higher order poles. In this case,
there is an evolution, not just in \textquotedblleft time\textquotedblright%
\ but also in \textquotedblleft space.\textquotedblright\

Already for the case of $N=1$ where we have a single M5-brane and an IR\ free
theory, the structure of simple punctures generated by this dynamical system
is surprisingly complex. To illustrate, we study in detail this class of
solutions, obtaining a full classification of initial conditions which produce a puncture.
The method also leads to a new class of punctures
for interacting 6D\ SCFT. For example, since we can also embed these solutions
inside large $N\times N$ matrices, we automatically generate $1/4$
BPS\ punctures for class $\mathcal{S}$ theories of $A_{N-1}$ type. For the higher rank
$N>1$ class $\mathcal{S}_{\Gamma}$ theories, these solutions also embed as solutions
which are diagonal block by block.

For a single M5-brane, namely $N=1$, the form of the dynamical system is more explicit and can be
fully automated. Introducing a position coordinate%
\begin{equation}
x(i)=\left\vert q(i)\right\vert ^{2}-\left\vert \widetilde{q}(i)\right\vert
^{2}\text{,}%
\end{equation}
the evolution from one value of position $x(i)$ and momentum $p(i)$ to the
next is:%
\begin{equation}
\left[
\begin{array}
[c]{c}%
p(i+1)\\
x(i+1)
\end{array}
\right]  =\left[
\begin{array}
[c]{cc}%
1 & 0\\
1 & 1
\end{array}
\right]  \left[
\begin{array}
[c]{c}%
p(i)\\
x(i)
\end{array}
\right]  -\left[
\begin{array}
[c]{r}%
\text{sgn }x(i)\\
\text{sgn }x(i)
\end{array}
\right]  ,
\end{equation}
where sgn$(x)=x/\left\vert x\right\vert $ for $x\ne0$ and 0 otherwise. This
evolution holds provided $x(i) \neq 0$. When $x(i_{\ast})=0$ vanishes at some
timestep, we can restart the evolution with another choice of initial condition.

Compared with other discrete dynamical systems, the appearance of the term
sgn$(x(i))$ leads to significant subtleties which are often only apparent in
sufficiently long quivers. Among other things, the presence of this term
obstructs a straightforward continuum limit, and small perturbations lead to
quite different behavior at later time steps.
Our goal will be to determine which initial conditions $p(1)$
and $x(1)$ actually produce viable punctures.

We show that a puncture corresponds to the
special condition that the orbit eventually returns to itself, namely we have
a periodic orbit. Moreover, periodic orbits are only achieved when the
initial momentum $p(1)$ is a rational number, which in turn means that all
other momenta are also rational and can all be expressed with same choice of denominator.
So, we can write $p(i)=a(i)/b$ with $a(i)$ and $b$ relatively prime integers.
We prove that:
\begin{equation}
b\neq0\text{ mod }4 \Rightarrow \text{Periodic Orbit}.
\end{equation}
Based on extensive numerical tests, we also find that the converse holds, namely when $b = 0 \text{ mod } 4$,
that we never have a periodic orbit. In the latter case, the late time values of the $x(i)$ become unbounded.

Returning to the system of physical punctures, we also see that one can
consider \textquotedblleft eternal punctures\textquotedblright\ where no
$x(i)$ vanish, and \textquotedblleft terminal punctures\textquotedblright\ in
which some of the $x(i)$ vanish. In the former case, there is really a single
dynamical system which continues for all time, while in the latter case, we
can restart the dynamical system with a new initial condition after each value
of $i$ for which $x(i)$ vanishes.\footnote{Note that in the dynamical system,
we are of course free to simply continue past this point without restarting
with a new initial condition.} In terms of the dynamical system, the two
classes of punctures have initial conditions:%
\begin{align}
\text{Eternal Punctures}  &  \text{: \ \ }x(1)\in%
\mathbb{R}
\text{ and }x(1)\neq p(1)=\frac{a(1)}{b}\text{ with }b\neq0\text{ mod }4\\
\text{Terminal Punctures}  &  \text{: \ \ }x(1)=p(1)=\frac{a(1)}{b}\text{ with
}b\neq0\text{ mod }4.
\end{align}
An interesting feature of the eternal punctures which does not appear for
$1/2$ BPS\ punctures of class $\mathcal{S}$ theories or for previously studied
punctures of $\mathcal{N} = (1,0)$ theories is the appearance of an overall zero mode: For
sufficiently small $\delta x$, the initial values for the position $x(1)$ and
$x(1)+\delta x$ produce the same eternal puncture. In the equations defined
by a puncture, this is also accompanied by a complex phase which cannot be removed by a unitary change
of basis. Note that the length of the orbit $k$ in an eternal puncture depends on $x(1)$ and $p(1)$,
so we can also write $k(x(1),p(1))$. Alternatively, we can hold fixed $p(1)$, and
for the set of $k$'s which appear, there is an interval of admissible $x(1)$'s which correspond
to the same puncture.

We can reach all periodic orbits by gluing together terminal
punctures, and then perturbing these solutions by the small parameter
$\delta x$. From this perspective, the
rather special case where we embed $su(2)$ in $su(N)$ corresponds to the
restricted case where $b$ is one or two. This illustrates the vast increase
in possible punctures for 4D $\mathcal{N}=1$ vacua.

The dynamical system we encounter is also of interest in its own right, and we
also analyze some additional aspects of its behavior, including sensitivity to
initial conditions and number of orbits as a function of orbit length. There
are fairly regular patterns for many of these quantities, indicating
additional non-trivial structure.

This non-trivial structure also persists for punctures with higher order poles, where there is
again a recursion relation, this time from one pole order to the next. Once
we pass beyond simple order poles, the evolution descends to a linear set of first order discrete evolution
equations rather than the second order behavior exhibited by the first order poles.

An important feature of the puncture solutions thus obtained is that they also generate punctures for more general
6D SCFTs. For example, in the case of $N$ M5-branes probing an A-type singularity, the scalars $q$, $\widetilde{q}$ and $\Sigma$
appearing in the puncture equations are promoted to $N \times N$ matrices. The particular ansatz where all such matrices are diagonal yields
a large class of new punctures, each governed by a decoupled dynamical system. More ambitiously, we can also consider dynamical systems
involving the full matrix structure. In this case, the appearance of a timelike evolution across quiver nodes still applies, though is
more involved. Similar considerations clearly apply for punctures of other 6D SCFTs.

The rest of this paper is organized as follows. In section \ref{sec:DYNAMO} we
establish the basic connection between punctures and dynamical systems.
We next turn to the classification of initial conditions for the dynamical system with $N=1$,
first studying the case with all $x(i)$ non-zero in section \ref{sec:ETERNAL} and
then turning to the case where some $x(i)$ vanish in section \ref{sec:TERMINAL}.
We analyze the structure of higher order poles in section \ref{sec:POLES}.
Section \ref{sec:HIGHER} discusses generalizations to the case of higher rank.
We present our conclusions and potential future directions in section \ref{sec:CONC}.
In Appendix \ref{app:DOMINO} we present some formal proofs establishing
which initial conditions of the dynamical system yield a periodic orbit.

\section{Punctures and Dynamical Systems \label{sec:DYNAMO}}

In this section we establish a correspondence between punctures
and dynamical systems. In doing so we reduce the problem of classification
to determining which initial conditions yield a puncture.

Recall first that the $\mathcal{S}_{k}$ theories
are given by $N$ M5-branes probing the transverse
geometry $\mathbb{R}_{\bot}\times\mathbb{C}^{2}/%
\mathbb{Z}
_{k}$. Punctures for such theories are obtained by dimensionally reducing on a long
cylinder $S^{1}\times\mathbb{R}_{\geq0}$. The reduction along the circle
factor of the cylinder geometry produces a 5D $\mathcal{N}= 1$ gauge theory
with scalar degrees of freedom given by $N\times N$ matrices of fields,
$\Sigma$, $Q$ and $\widetilde{Q}$, and we impose supersymmetry preserving
boundary conditions on the semi-infinite interval. There is then a power series
expansion for modes near the boundary $(t= 0)$ of $\mathbb{R}_{\geq0}$:
\begin{equation}
Q=\sum\limits_{n>0}\frac{Q_{n}}{t_{n}}\,,\quad\widetilde{Q}=\sum
\limits_{n>0}\frac{\widetilde{Q}_{n}}{t^{n}}\quad\text{and}\quad\Sigma
=\sum\limits_{n>0}\frac{\Sigma_{n}}{t^{n}},
\end{equation}
and the puncture is governed by the matrix equations \cite{Heckman:2016xdl}:
\begin{align}
\sum\limits_{k+l=m}[\Sigma_{k},Q_{l}]  &  =(m-1)Q_{m-1} & \sum\limits_{k+l=m}%
[Q_{k},\widetilde{Q}_{l}]  &  =0\nonumber\\
\sum\limits_{k+l=m}[\Sigma_{k},\widetilde{Q}_{l}]  &  =(m-1)\widetilde{Q}%
_{m-1} & \sum\limits_{k+l=m}[Q_{k},Q_{l}^{\dagger}]+[\widetilde{Q}%
_{k},\widetilde{Q}_{l}^{\dagger}]  &  =(m-1)\Sigma_{m-1},\label{higherpoleconstr}
\end{align}
here, the subscript on each matrix denotes the order of the pole.

Proceeding order by order in the poles, we see that for first order poles,
i.e., $m=1$, we have the quadratic equations:%
\begin{equation}
\lbrack\Sigma_{1},Q_{1}]=Q_{1}\text{, \ \ }[\Sigma_{1},\widetilde{Q}%
_{1}]=\widetilde{Q}_{1}\text{, \ \ }[Q_{1},\widetilde{Q}_{1}]=0\text{,
\ \ }\Sigma_{1}=[Q_{1},Q_{1}^{\dag}]+[\widetilde{Q}_{1},\widetilde{Q}%
_{1}^{\dag}].
\end{equation}
Since we shall mainly focus on the first order poles, we shall drop the
subscript for the $\Sigma_{1}$, $Q_{1}$ and $\widetilde{Q}_{1}$, and simply
write $\Sigma$, $Q$ and $\widetilde{Q}$, respectively.

Observe that at second order and above, the equations become linear in the
corresponding pole order. In this sense, once we solve the first order
equations, all subsequent orders can be iteratively solved by linear
transformations. Therefore, we can \textquotedblleft
evolve\textquotedblright\ from one pole order to the next in the generic
case and it is enough to focus on the first order poles. Note that there
are also some special cases where the series of higher poles stops at some
point. We will discuss them in detail in section~\ref{sec:POLES}.

Let us make a few comments on the space of solutions. First, we note that
there is an obvious redundancy by a unitary change of basis. Two punctures
lead to the same physical theory under transformations of the form:%
\begin{equation}
\Sigma\mapsto U^{\dag}\Sigma U\text{, \ \ }Q\mapsto U^{\dag}QU\text{,
\ \ }\widetilde{Q}^{\dag}\mapsto U^{\dag}\widetilde{Q}^{\dag}U,
\end{equation}
for $U$ a unitary matrix in $U(N)$. Note also that the pair $Q$ and
$\widetilde{Q}^{\dag}$ rotate as a doublet under the $SU(2)$ R-symmetry of the
parent 5D theory. This is broken to an $\mathcal{N}=1$ subalgebra by the
presence of the puncture. Counting up the total number of degrees of freedom,
we have $5N^{2}$ real degrees of freedom, subject to $4N^{2}$ constraint
equations, and $N^{2}$ \textquotedblleft gauge redundancies.\textquotedblright
From this perspective, one might expect to only find a discrete
point set of solutions, and in many cases this is indeed correct.
The caveat to this is that sometimes there can
be residual zero modes.''

Let us now further specialize to the case of first order poles for the class
$\mathcal{S}_{k}$ theories. Working in terms of $Nk\times Nk$ matrices, we
decompose into blocks, each of which is an $N\times N$ matrix. Applying the
orbifold projection of reference \cite{Douglas:1996sw}, the surviving entries of the matrices
in a \textquotedblleft quiver basis\textquotedblright\ are given by a set of
linear maps between $k$ different $N$-dimensional vector spaces $V_{1}%
,...,V_{k}$:%
\begin{equation}
p(i):V_{i}\rightarrow V_{i}\text{, \ \ }q(i):V_{i+1}\rightarrow V_{i}\text{,
\ \ }\widetilde{q}(i):V_{i}\rightarrow V_{i+1}.
\end{equation}
Embedding in the original $Nk\times Nk$ matrices, we can also write:%
\begin{align}\label{SIGMA1}
\Sigma_{1}  &  =\left[
\begin{array}
[c]{cccc}%
p(1) &  &  & \\
& \ddots &  & \\
&  & p(k-1) & \\
&  &  & p(k)
\end{array}
\right]  \\
Q_{1}  &  =\left[
\begin{array}
[c]{cccc}
& q(1) &  & \\
&  & \ddots & \\
&  &  & q(k-1)\\
q(k) &  &  &
\end{array}
\right] \\
\widetilde{Q}_{1}  &  =\left[
\begin{array}
[c]{cccc}
&  &  & \widetilde{q}(k)\\
\widetilde{q}(1) &  &  & \\
& \ddots &  & \\
&  & \widetilde{q}(k-1) &
\end{array}
\right]  . \label{Qtilde1}
\end{align}
The puncture equations now reduce to:%
\begin{align}
q(i) &= p(i)q(i)-q(i)p(i+1)\\
-\widetilde{q}(i) &= \widetilde{q}(i)p(i)-p(i+1)\widetilde{q}(i) \\
0 &= q(i)\widetilde{q}(i)-\widetilde{q}(i-1)q(i-1)  \\
p(i) &= \left[  q(i)q^{\dag}(i)-\widetilde{q}^{\dag}(i)\widetilde{q}(i)\right]
-\left[  q^{\dag}(i-1)q(i-1)-\widetilde{q}(i-1)\widetilde{q}^{\dag
}(i-1)\right] .
\end{align}
Again, there is a redundancy in the solutions given by acting by a set of
unitary matrices:%
\begin{equation}
p(i)\mapsto U^{\dag}(i)p(i)U(i)\text{, \ \ }q(i)\mapsto U^{\dag}%
(i)q(i)U(i+1)\text{, \ \ }\widetilde{q}^{\dag}\mapsto U^{\dag}(i)\widetilde{q}U(i+1),
\end{equation}
for $U(i)$ a unitary transformation of the vector space $V_{i}$. An important feature of this
set of equations is that it is recursive in structure, namely if we possess a solution, we can iteratively solve for the
matrices at node $i+1$ using the matrices at node $i$. In this sense, punctures always generate a dynamical system.
Of course, it may prove difficult to explicitly construct solutions.

Our plan in the remainder of this section will be to study the formal
structure of these puncture equations. The analytically
most tractable case is $N=1$, namely a single M5-brane
probing an orbifold singularity. This already leads to a surprisingly rich
structure, and one of our aims in the remainder of this paper will be to fully
characterize the solutions in this special case.

\subsection{Specialization to $N=1$}

To gain further understanding of the possible solutions to this system of
equations, we now specialize even further to the case of $N=1$. In this case,
the matrices at each time step are just numbers, and we can write the whole
system as:%
\begin{align}
q(i) &= \left[  p(i)-p(i+1)\right]  q(i) \\
-\widetilde{q}(i) &= \left[  p(i)-p(i+1)\right]  \widetilde{q}(i)  \\
0 &= q(i)\widetilde{q}(i)-\widetilde{q}(i-1)q(i-1)  \label{qqtcomm}\\
p(i) &= \left[  \left\vert q(i)\right\vert ^{2}-\left\vert \widetilde{q}(i)\right\vert
^{2}\right]  -\left[  \left\vert q(i-1)\right\vert ^{2}-\left\vert
\widetilde{q}(i-1)\right\vert ^{2}\right] .
\end{align}
For any time step $i$, at most one of $q(i)$ or
$\widetilde{q}(i)$ can be non-zero (and thus line (\ref{qqtcomm}) is always
satisfied). Assuming we are in the generic situation where at least one is
non-zero, we fully characterize the solution by the recursion relations:%
\begin{equation}
\left[
\begin{array}
[c]{c}%
p(i+1)\\
x(i+1)
\end{array}
\right]  =\left[
\begin{array}
[c]{cc}%
1 & 0\\
1 & 1
\end{array}
\right]  \left[
\begin{array}
[c]{c}%
p(i)\\
x(i)
\end{array}
\right]  -\left[
\begin{array}
[c]{r}%
\text{sgn }x(i)\\
\text{sgn }x(i)
\end{array}
\right]  , \label{evolve}%
\end{equation}
where:%
\begin{equation}
x(i)=\left\vert q(i)\right\vert ^{2}-\left\vert \widetilde{q}(i)\right\vert
^{2}\text{.}%
\end{equation}
One can also work purely in terms of a second order relation:%
\begin{equation}
-x(i-1)+2x(i)-x(i+1)=\text{sgn }x(i),
\end{equation}
so one can visualize this as a discrete wave equation subject to a non-trivial source.

To a certain extent, even the particular values of the $x(i)$ are redundant.
The only combinatorial data which actually enters in specifying the puncture
is the sign of $x(i)$. We can therefore introduce a spin variable
$s(i) = \mathrm{sgn}(x(i))$ with values $\pm1$ or $0$ depending on whether $x(i)$ is
positive, negative, or vanishes.

Now, from the perspective of finding puncture solutions, if we ever encounter
a value $i_{\ast}$ such that $x(i_{\ast}) = 0$, we can simply supply another
initial condition to the dynamical system and continue evolving. Said
differently, a sequence of non-zero values of $x(i)$ determines a puncture,
and we can of course produce another puncture by appending to it another such
sequence by restarting the dynamical system with a different choice of initial conditions.

Along these lines, we refer to a \textquotedblleft terminal
puncture\textquotedblright\ as one for which $x(k)=0$. Additionally, there are
punctures for which no value of $x(i)$ vanishes. In such situations, the fact
that both parent matrices $Q$ and $\widetilde{Q}$ are nilpotent means that
there is at least one sign flip in the sequence of $x(i)$'s. We refer to this
as an \textquotedblleft eternal puncture.\textquotedblright\ Summarizing, we
have two types of punctures to consider:%
\begin{align}
\text{Terminal Puncture}\text{: }  &  x(0)=x(k)=0\\
\text{Eternal Puncture}\text{: }  &  x(i)\neq 0\text{ for all }i=1,...,k.
\end{align}
Note that whereas for a puncture we naturally identify $x(0)$ and $x(k)$, in
the context of a dynamical system, nothing forces us to do so.

Indeed, from the perspective of the dynamical system we can just study the
evolution equations of line (\ref{evolve}); we simply
feed in a choice of initial conditions at $i=0$, namely $p(0)$ and
$x(0)$, and then proceed to evolve it for all time steps $i\in%
\mathbb{Z}
_{\geq0}$. In the process of this evolution, it can happen that $x(i)$
vanishes for several choices of $i$, and we can label this subsequence of
values as $k_{1},...,k_{m},...$. Each such termination leads to a valid
solution to the puncture equations for some choice of quiver size (typically
different from $k$). In the context of the dynamical system, however, we can
continue to evolve past this point of vanishing. It could even happen that the
sequence hits zero at some value, but does not (yet) repeat its profile through
the phase space.\footnote{A consequence of our analysis, however, is that if an orbit
for $x(i)$ passes through zero twice then it necessarily executes a periodic orbit.}

\subsection{Periodic Orbits}

Even so, the only initial conditions we need concern ourselves with are those
with a periodic orbit, though a priori the period length may differ from $k$.
In the case where no $x(i)$ vanish for the puncture equations, namely we have
an eternal puncture, it is clear that for us to get a solution to the puncture
equations, the dynamical system must execute a periodic orbit. In the case of
a terminal puncture, we recall that necessarily, such a puncture has $x(0)=0$
and $x(k)=0$. In terms of the $x(i) $,\ recall that we can express the
puncture equations as a second order discrete difference:%
\begin{equation}
-x(i-1)+2x(i)-x(i+1)=\text{sgn }x(i).
\end{equation}
Summing from $i=k-n+1$ to $i=k+n-1$ for some $n\geq1$ yields a telescoping
series, so we obtain:%
\begin{equation}
-x(k-n)-x(k+n)=\underset{i=k-n+1}{\overset{k+n-1}{\sum}}\text{sgn }x(i).
\label{telescoped}%
\end{equation}
Proceeding by induction, we see that for $n=1$, we need to evaluate sgn
$x(k)=0$, so we learn that $x(k-1)=-x(k+1)$. \ Assuming the relation:%
\begin{equation}
x(k-n)+x(k+n)=0 \label{updown}%
\end{equation}
holds for $n=1,...,N$, it clearly also holds for $n=N+1$, since the signs on
the righthand side of line (\ref{telescoped}) cancel pairwise. This
establishes relation (\ref{telescoped}) for all $n$. Since we essentially just
run the evolution in reverse as we cross through a zero, we see that in the
case of a puncture, where $x(0)=x(k)=0$, the orbit necessarily repeats after
at most $2k$ steps. That is to say, we can set $n=k-i$ to obtain:%
\begin{equation}
x(i)=-x(2k-i)=x(2k+i)\text{.}%
\end{equation}
where in the second equality, we used the fact that $x(2k)=0$, so
$x(2k-i)=-x(2k+i)$.

This establishes that for the analysis of punctures, it is enough to focus on
periodic orbits of the dynamical system. Additionally, we see that even if
$x(0)=x(k)=0$, the length of the orbit may be $2k$ rather than $k$. We
reference these two possibilities as \textquotedblleft short
orbits\textquotedblright\ and \textquotedblleft long orbits\textquotedblright:%
\begin{align}
\text{Short Orbit}\text{: Period of Length }  &  k\\
\text{Long Orbit}\text{: Period of Length }  &  2k.
\end{align}
Note that eternal punctures always come from short orbits, whereas a terminal
puncture could be either a short or long orbit.

Now, because we have a trajectory which repeats after at most $2k$ steps, we
see at once that the momenta must be quantized in units of $1/k$ or $1/2k$.
The only case which corresponds to representations of $\mathfrak{su}(2)$ is the very
special case where the momenta is a half integer or an integer.

We now turn to the the initial conditions necessary for our dynamical system
to execute a periodic orbit.

\section{Eternal Punctures \label{sec:ETERNAL}}

Our aim in this section will be to determine initial conditions which produce eternal
puncture. A helpful feature of this case is that there is a natural geometric
interpretation of the solutions. To see this, we return to the
large matrices $Q$ and $\widetilde{Q}$, which are interpreted as the matrix
collective coordinates for branes moving on the space $\mathbb{C}^{2}/\Gamma$.
Now, although these matrices are nilpotent, we see that for an eternal
puncture, the combinations:%
\begin{equation}
Q_{\pm}=Q\pm\widetilde{Q}^{\dag}%
\end{equation}
have non-vanishing determinant. Indeed, note that since:%
\begin{equation}
\text{Tr}(Q_{\pm}^{j})=0\text{ \ \ for \ \ }1\leq j<k,
\end{equation}
we also learn that the eigenvalues of $Q_{+}$ and $Q_{-}$ are:%
\begin{equation}
\text{Eigen}(Q_{+})=\text{Eigen}(Q_{-})=\left\{  \zeta\times
\underset{i=1}{\overset{k}{\prod}}(q(i) \pm\widetilde{q}^{\dag}(i))\text{
\ \ s.t. \ \ }\zeta^{k}=1\right\}  .
\end{equation}
Here, we have used the fact that there must be an even number of sign flips in
the $x(i)$, since the momentum needs to return to its initial value after $k$
time steps.

Even so, we must exercise some caution with our geometric interpretation
because the coordinates $Q_{+}$ and $Q_{-}$ do not commute. There is a sense
in which we have a semi-classical limit, however, because we can consider the
special limit where the rank of $[Q_{+},Q_{-}]$ is much smaller than $k$. In
this case, the relative number of total sign flips is also quite small.

The rest of this section is organized as follows. First, we show that for an
orbit of length $k$, the momenta is quantized in units of $1/k$. Additionally,
we show that in this case, there is no quantization in the value of the
$x(i)$'s, and in fact, that there is an overall zero mode for the system. We
then turn to a detailed analysis of the dynamical system associated with
eternal punctures, showing in particular that for nearly all values of the
momenta, we do indeed obtain a periodic orbit. There are, however, a few cases
which do not appear to exhibit this structure, corresponding to the special
cases where $p(1)=a/b$ for $b = 0$ mod $4$.

\subsection{Rational Momenta}

The first non-trivial observation we can make is that the momenta of the
system necessarily take values in the rational numbers. To see this, consider
again the $k$th iteration of the dynamical system:%
\begin{align}
p(k+1)  &  =p(1)-\underset{i=1}{\overset{k}{\sum}}\text{sgn }x(i)\\
x(k+1)  &  =kp(1)+x(1)-\underset{i=1}{\overset{k}{\sum}}(k-i)\text{sgn }x(i).
\end{align}
On the other hand, since we have assumed $x(k+1)=x(1)$, the second line
already tells us that $p(1)$ is a rational number such that $kp(1)$ is an
integer:%
\begin{equation}
p(1)=\frac{1}{k}\underset{i=1}{\overset{k}{\sum}}(k-i)\text{sgn }x(i).
\end{equation}
The challenge, of course, is that we do not a priori know which signs for
$x(i)$ will yield a consistent solution to the dynamical system. Note also
that there is a priori no reason for the $x(i)$ to take values in the rational
numbers.

\subsection{Zero Mode}

Indeed, the dynamical system leaves unfixed an overall \textquotedblleft zero
mode.\textquotedblright\ Suppose
that we have managed to find a consistent choice of initial conditions for
$x(1)$ and $p(1)$ which solves the conditions of the dynamical system. Now,
for $\delta x$ sufficiently small, we can perturb each of the $x(i)$ so that
we do not change the sign of any $x(i)$:%
\begin{equation}
x(i)\mapsto x(i)+\delta x\text{ \ \ with \ \ sgn }x(i)=\text{sgn}(x(i)+\delta
x). \label{perto}%
\end{equation}
Observe that because only the differences of $x(i)$ show up in this system,
this \textquotedblleft zero mode\textquotedblright\ in fact decouples. So in
general, there is a small continuous modulus associated with eternal punctures.
Additionally, we see that whereas the momentum is naturally quantized, the
position can in principle take on arbitrary real values.

An additional comment is that there is also an unfixed overall phase for the
$q(i)$ and $\widetilde{q}^{\dag}(i)$. By a
choice of unitary transformation, we can map each $q(i)+\widetilde{q}^{\dag}(i)$ to:%
\begin{equation}
\left(  q(i)+\widetilde{q}^{\dag}(i)\right)  \mapsto U(i)\left(
q(i)+\widetilde{q}^{\dag}(i)\right)  U^{\dag}(i+1)
\end{equation}
for some choice of complex phases $U(i)$. Doing so, we can eliminate most of
the phases using the constraints:
\begin{equation}
\arg(U(i)U^{\dag}(i+1))=-\arg(q(i)+\widetilde{q}^{\dag}(i)),
\end{equation}
so there is one overall phase which cannot be eliminated by a change of basis.
This is in some sense the \textquotedblleft bosonic partner\textquotedblright%
\ to the radial mode $\delta x$, as one would expect in an
$\mathcal{N}=1$ supersymmetric theory. Whereas the complex
phase of the $q$ and $\widetilde{q}$'s completely decouples from our analysis,
we will see that suitable tuning of $\delta x$ allows us to interpolate from
eternal punctures to terminal punctures. For higher values of $N$, we
anticipate that in general, the corresponding \textquotedblleft zero
modes\textquotedblright\ can also mix, so that the matrices cease to remain
diagonal.

By inspection, the eigenvalues are arranged along a circle, as befits the
interpretation in terms of image branes. In this picture, the perturbation of
line (\ref{perto}) corresponds to moving each image in or out. Similar
considerations hold from analyzing the \textquotedblleft radius
squared\textquotedblright\ obtained from Tr$(Q_{+}^{\dag}Q_{+}+Q_{-}^{\dag
}Q_{-})$.

In the low energy effective field theory, the fluctuation $\delta x$ signals the presence of a 
free chiral multiplet, which we associate with a corresponding Goldstone mode. Note that the field range of 
this mode is limited to small fluctuations. Indeed, as we increase the size of $\delta x$, we can 
jump from one value of $k$ to another value. In the stringy geometry, this completely changes the compactification 
geometry. In the low energy effective field theory specified by the puncture, this value of $k$ shows up 
in the spectrum of defect / line operators, as per reference \cite{DelZotto:2015isa}. From this perspective, 
we interpret possible jumping behavior in the 4D vacua as a phenomenon akin to skyrmionic excitations. We leave a 
complete treatment of this interesting phenomenon for future work.

\subsection{Minimal Momenta and Maximal Positions}

Much as in other classical systems, it is helpful to analyze the special cases
where $\left\vert p(i)\right\vert $ is as small as possible. Unlike a system
with continuous time steps, however, this minimal value need not be zero. As
it is helpful in characterizing other aspects of our solutions, we now study
in detail the minimal momenta obtained in the orbit of an eternal puncture.
The general claim we make is that in the course of its evolution, the momenta
will always pass to a small value, $p$:%
\begin{equation}
\left\vert p\right\vert \leq1/2.
\end{equation}
Moreover, the corresponding value of the position leads to an approximation of
the maximum which would be obtained in the continuum limit of infinitesimal
time steps.

To see why this always occurs, it is enough to consider the special case where
$p(1)>0$ and $x(1)>0$. Indeed, if $p(1)<0$ and $x(1)<0$, the same argument
will apply, and in the case where $x(1)$ and $p(1)$ have opposite sign, we
observe that by evolving the system for a sufficient number of time steps,
they eventually have the same sign anyway.

Consider, then, the case where $p(1)>0$ and $x(1)>0$. In this case, we proceed
for some number of time steps until the sign of $x(i)$ changes from positive
to negative. To determine where this occurs, write $p(1)$ as:%
\begin{equation}
p(1)=p+m,
\end{equation}
where $m\in%
\mathbb{Z}
_{\geq0}$ and $\left\vert p\right\vert \leq1/2$. If $m=0$, then there is
nothing to show, so we assume to the contrary that $m>0$. Now, after the first
time step, the new values for our dynamical system are:%
\begin{align}
p(2)  &  =p+m-\text{sgn }x(1)=p+m-1\\
x(2)  &  =x(1)+p(1)-\text{sgn }x(1)=x(1)+p+m-1\text{,}%
\end{align}
so again, $x(2)>0$. Thus, there will be a sequence of $+$ signs for $x(j)$,
and eventually there will be a sign flip at some later value of $j$.
Continuing in this way, we seek out the largest value of $j$ such that:%
\begin{equation}
x(j)>0\text{ \ \ for \ \ }1\leq j\leq j_{\ast}\text{ \ \ and \ \ }x(j_{\ast
}+1)<0.
\end{equation}
Iterating the dynamical system $j$ times in this range, we have, by
assumption, that:%
\begin{equation}
x(j+1)=x(1)+jp(1)-\frac{j(j+1)}{2}=x(1)+j(p+m)-\frac{j(j+1)}{2}.
\end{equation}
This is a quadratic polynomial in $j$, and its zeros occur at:%
\begin{equation}
j_{\pm}=\left(  p+m-\frac{1}{2}\right)  \pm\sqrt{\left(
p+m-\frac{1}{2}\right)  ^{2}+2x(1)},
\end{equation}
which in general is not an integer. The first integer value of $j$ which
yields a negative value, namely $x(j_{\text{flip}})<0$ is then given by
rounding up using the ceiling function:%
\begin{equation}
j_{\text{flip}}=\text{Ceil}\left[  \left(  p+m-\frac{1}%
{2}\right)  +\sqrt{\left(  p+m-\frac{1}{2}\right)  ^{2}%
+2x(1)}\right]  . \label{jflip}%
\end{equation}
We can also establish a crude lower bound for $j_{\text{flip}}$ since $x(1)>0
$:%
\begin{equation}
j_{\text{flip}}\geq2p+2m-1,
\end{equation}
so in other words, the integral part of $p(1)$ is bounded above by:%
\begin{equation}
m\leq p+\frac{j_{\text{flip}}+1}{2}.
\end{equation}
Due to this, we see that $p(i)$ can indeed decrease for $m$ steps whilst
$x(i)$ still remains positive. At the $(m+1)$th step, the value of $p(m+1)$
and $x(m+1)$ is therefore:%
\begin{align}
p(m+1)  &  =p\\
x(m+1)  &  =x(1)+\frac{m^{2}}{2}+m\left(  p-\frac{1}{2}\right)  .
\end{align}
We also see that this minimal value of $p(i)$ coincides with a maximal value
of $x(i)$, much as one would expect in extremizing a continuous function. To
see this, note first that $x(i)$ is strictly increasing as we move from
$i=1,...,m+1$. Additionally, at $i=m+2$, we have that $x(m+1)$ is still
positive, but that:%
\begin{align}
p(m+2)  &  =p-\frac{1}{2}\\
x(m+2)  &  =x(m+1)+(p-1),
\end{align}
so $x(m+2)<x(m+1)$, since $(p-1)<0$. See figure~\ref{fig:curves}
for a depiction of this local maximum.
\begin{figure}
  \centering
  \includegraphics{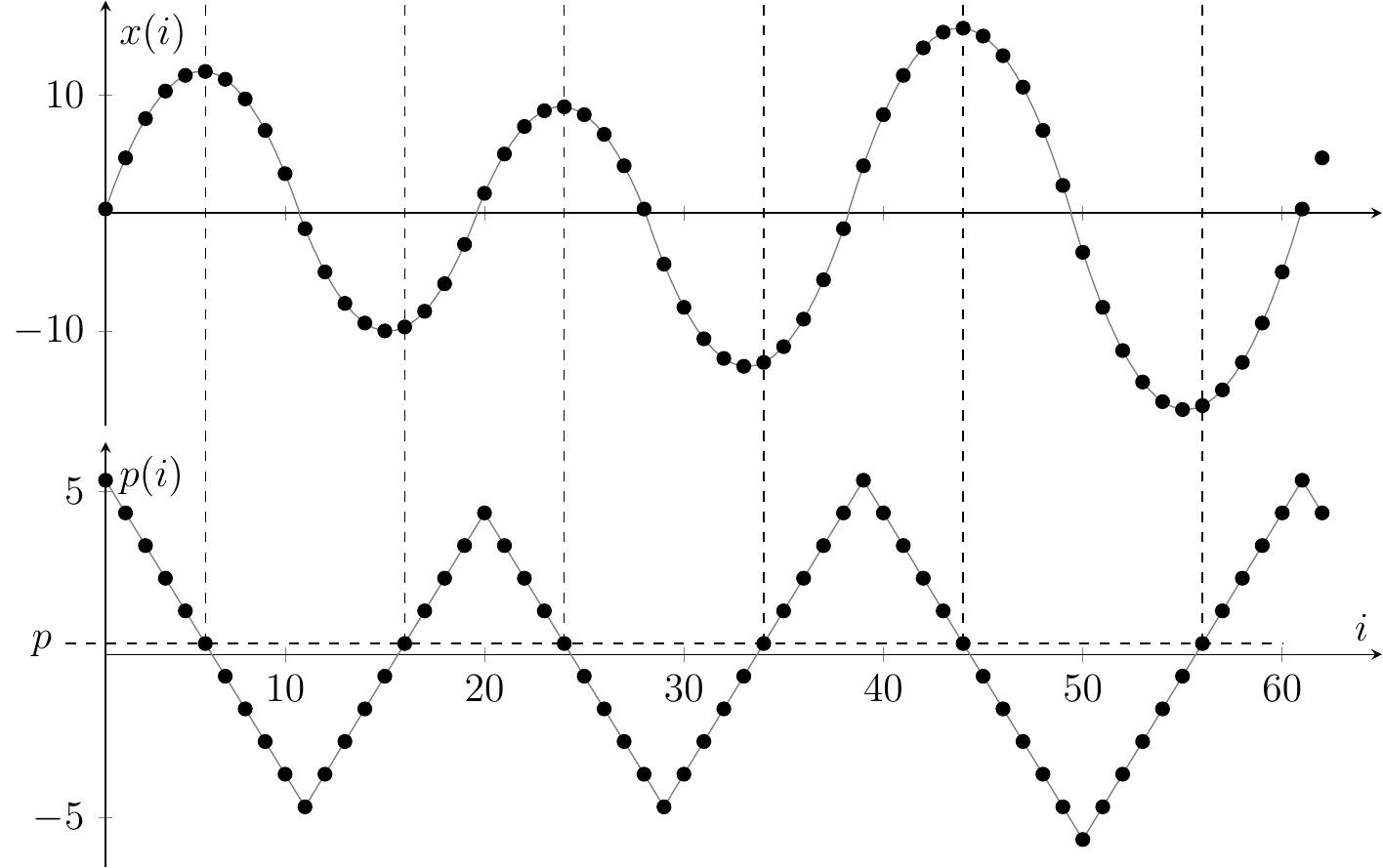}
  \caption{Approximated maxima and minima of $x(i)$ arise for all $i$ where
$p(i)=p$ holds. In the plotted example, we have $x(1) = 1/3$, and $p(1) = x(1) + 5$, with $p = 1/3$. For this case, we realize an eternal puncture where no $x(i)$ vanishes, and the total length of the orbit is $k = 60$.}\label{fig:curves}
\end{figure}
Given this, it is also natural to ask where the next local extremum will
occur. Indeed, after reaching the \textquotedblleft top of the
hill\textquotedblright, we see that the $x(i)$ start to decrease until
eventually we pass to a negative value of $x(j_{\text{flip}})$ where the
position flips sign:%
\begin{align}
p(j_{\text{flip}})  &  =p(1)-(j_{\text{flip}}-1)\label{pjflip}\\
x(j_{\text{flip}})  &  =x(1)+(j_{\text{flip}}-1)p(1)-\frac{j_{\text{flip}%
}(j_{\text{flip}}-1)}{2}. \label{xjflip}%
\end{align}
The $x(i)$ for $i\geq j_{\text{flip}}$ then remain negative for some
additional number of time steps.

Returning to our original argument, we also see that $p(j_{\text{flip}})$ can
be written as:%
\begin{equation}
p(j_{\text{flip}}) =p+\left(  m-(j_{\text{flip}%
}-1)\right)  ,
\end{equation}
namely the term in parentheses is a negative integer, and $p$ is
again a rational number between $-1/2$ and $+1/2$. Continuing with our
previous discussion, we know that since $p(j_{\text{flip}})$ and
$x(j_{\text{flip}})$ are both negative, the values of $x(i)$ will now be
negative for a while, and will eventually reach a minimal value for $p(i)$
where the momentum reaches its smallest norm, namely $p$. Letting
$i_{\text{max}}$ denote the local value of the first \textquotedblleft
maximum\textquotedblright\ where $x(i_{\text{max}})>0$ and $p(i_{\text{max}%
})=p$, and $i_{\text{min}}$ denote the first minimum where
$x(i_{\text{min}})<0$ and $p(i_{\text{min}})=p$, we see that:%
\begin{equation}
p(i_{\text{min}})=p(i_{\text{max}})-\underset{i_{\text{max}}\leq j\leq
i_{\text{min}}-1}{\sum}\text{sgn }x(j),
\end{equation}
so in particular, since the momentum is the same at the \textquotedblleft top
of the hill\textquotedblright\ and the \textquotedblleft bottom of the
hill\textquotedblright, namely $p(i_{\text{min}})=p(i_{\text{max}%
})=p$, we have an equal number of $+$ and $-$ signs for sgn
$x(j)$ in this range.

We can also calculate the number of time steps between each local maximum and
minimum. Indeed, since we now know that the trajectory hits an equal number of
$+$ and $-$ signs in passing between the two extrema, it is enough to start at
$i=i_{\text{max}}$ and evolve until the $x(i)$ flip sign. Doing so, we require
the existence of a positive integer $l_{+}$ with:%
\begin{align}
x(l_{+}+i_{\text{max}})  &  >0\\
x(l_{+}+i_{\text{max}}+1)  &  <0.
\end{align}
So, since we also know:%
\begin{equation}
x(l_{+}+i_{\text{max}}+1)=x(i_{\text{max}})+l_{+}p(i_{\text{max}})-\frac
{l_{+}(l_{+}+1)}{2},
\end{equation}
we can solve for the roots of this quadratic equation in $l_{+}$ to find the
corresponding value of $l_{+}$:%
\begin{equation}
l_{+}=\text{Floor}\left[  \left(  p-\frac{1}{2}\right)
+\sqrt{\left(  p-\frac{1}{2}\right)  ^{2}+2x(i_{\text{max}}%
)}\right]  .
\end{equation}
Going another $l_{+}$ steps, we reach a minimum:%
\begin{equation}
x(i_{\text{min}})=x(2l_{+}+i_{\text{max}}+1)=x(i_{\text{max}})+2l_{+}%
p(i_{\text{max}})-\left(  \frac{l_{+}}{2}\right)  ^{2}.
\end{equation}
We can also evaluate the total length of time taken in passing from the local
maximum to the local minimum:%
\begin{equation}
i_{\text{min}}-i_{\text{max}}=2\text{Ceil}\left[  \left(  p%
-\frac{1}{2}\right)  +\sqrt{\left(  p-\frac{1}{2}\right)
^{2}+2x(i_{\text{max}})}\right]  ,
\end{equation}
where we have used the fact that the number of $+$'s and $-$'s is exactly the
same in passing from the local maximum to the local minimum. Passing now from
the local minimum to the next maximum, we obtain a rather similar expression,
but where we now start with $x(i_{\text{min}})<0$, eventually reaching a local
maximum. Labelling the local extrema as $i=i_{\text{ext}}^{(1)},i_{\text{ext}%
}^{(2)},i_{\text{ext}}^{(3)},...$, in which $i_{\text{ext}}^{(\text{odd})}$
denotes a local maximum, and $i_{\text{ext}}^{(\text{even})}$ denotes a local
minimum for the position $x(i)$, the distance $l^{(m)}$ between subsequent
extrema is:%
\begin{equation}\label{l(m)}
l^{(m)}=i_{\text{ext}}^{(m+1)}-i_{\text{ext}}^{(m)}=2\text{Ceil}\left[
\left(  p\text{sgn }x(i_{\text{ext}}^{(m)})-\frac{1}{2}\right)
+\sqrt{\left(  p\text{sgn }x(i_{\text{ext}}^{(m)})-\frac{1}%
{2}\right)  ^{2}+2\left\vert x(i_{\text{ext}}^{(m)})\right\vert }\right] .%
\end{equation}

Since we are ultimately interested in determining periodic orbits of the
dynamical system, it is clearly fruitful to focus on this special subset of
values. The evolution from one extremum to the next is again controlled by a
recursion relation:%
\begin{equation}\label{iteration}
x(i_{\text{ext}}^{(m+1)})=x(i_{\text{ext}}^{(m)})+l^{(m)}p_{\text{min}%
}-\left(  \text{sgn }x(i_{\text{ext}}^{(m)})\right)  \times\left(
\frac{l^{(m)}}{2}\right)  ^{2}.
\end{equation}
Rather importantly, we also see that the value of the momentum is the same for
all of these special values:%
\begin{equation}
p=p(i_{\text{ext}}^{(1)})=p(i_{\text{ext}}^{(2)})=\dots \,.
\end{equation}
To find initial conditions for which we execute a periodic orbit, it is
therefore enough to find two values of extremal values for which the value of
$x(i)$ repeats, namely:%
\begin{equation}
x(j_{\text{ext}})=x(j_{\text{ext}}^{\prime}).
\end{equation}
Our aim in the following subsection will be to determine the conditions
necessary to achieve such a periodic orbit.

\subsection{Periodic Orbit Conditions}\label{porbitconds}

We would now like to understand in greater detail the requirement that our
orbit is actually periodic. The purpose of Appendix \ref{app:DOMINO}
is to formally prove:
\begin{equation}
p(1)  =\frac{a}{b}\text{ \ \ with \ \ }a,b\text{ coprime and }b \neq 0\,\text{mod}\,4 \Rightarrow \text{Periodic Orbit}
\end{equation}
When $b = 0$ mod $4$, we also find strong numerical evidence that the orbit
never closes. We have explicitly checked this up to orbit lengths $k \leq 2000$.

The main element of the proof is to study what happens when the $x(i)$ become very
large. Indeed, we know that for a periodic orbit there is always an upper bound
on all the $x(i)$, so an orbit which is not periodic must necessarily suffer from
unbounded growth in the $x(i)$. Thankfully, as $x(i)$ becomes very large, some of
the non-analytic behavior in the locations of sign flips is also reduced. The main
strategy of the proof presented in the Appendix is that in this regime, the effects
of the square root appearing in the length of a contiguous block of $+/-$'s disappears,
and a simplified dynamical system involving the ceiling function is all that remains. This
is still a challenging dynamical system to analyze, but the number theoretic aspects of
its behavior are amenable to an exact analysis.

To see how this comes about, we first present an analytic treatment for the time
evolution governed by \eqref{iteration}. In order to increase the readability of the
equations in this subsection, we use the abbreviations:
\begin{equation}
  y^{(m)} = x(i_\text{ext}^{(m)})
    \quad \text{and} \quad
  l = l^{(1)}\,.
\end{equation}
In addition to the evolution $y^{(m)} \rightarrow y^{(m+1)}$ in \eqref{iteration},
we also need its inverse $y^{(m)} \rightarrow y^{(m-1)}$ in the following
discussion. It is given by
\begin{align}\label{iterationinv}
  y^{(m-1)} &= y^{(m)}-\tilde l^{(m)} p - \left( \frac{ \tilde l^{(m)}}{2}\right) ^{2}
  \sgn ( y^{(m)} - p ) \quad \text{with}\\
  \tilde l^{(m)} &= 2\text{Ceil}\left[ - p \, \sgn ( y^{(m)} - p ) - \frac{1}{2} + \sqrt{\left(
    - p \, \sgn ( y^{(m)} - p ) - \frac12 \right) ^2 + 2 \left\vert y^{(m)} - p \right\vert } \right] \,. \label{tildel(m)}
\end{align}
Finally note, that the equations \eqref{iteration} and \eqref{l(m)} possess a
$\mathbb{Z}_2$ symmetry acting as
\begin{equation}
  y^{(m)} \rightarrow - y^{(m)}\,, \quad p \rightarrow - p\,, \quad l^{(m)} \rightarrow l^{(m)}\,.
\end{equation}
Therefore, it is sufficient to restrict the following discussion to
$0 < p < 1/2$. From the results for positive $p$, the time evolution for
negative $p$ follows immediately by flipping the sign of $y^{(m)}$.

In general, it not possible to find an analytic expression for all $y^{(m)}$
starting from an arbitrary initial value $y^{(1)}$. This is due to three
properties of the equations \eqref{l(m)} and \eqref{iteration}:
\begin{enumerate}
  \item Equation \eqref{l(m)} contains the highly non-linear ceiling function
  \item Equation \eqref{l(m)} contains a square root
  \item Equation \eqref{iteration} is non-linear in $l^{(m)}$\,.
\end{enumerate}
Item one captures a fundamental feature of the dynamical system. Hence, we
can not get rid of the ceiling function. But, we can deal with the two other
points, if we impose some additional restrictions on $y^{(m)}$. To see how
this works, we first switch from $y^{(m)}$ to the more adapted variable
$\Delta^{(m)}$. It is implicitly defined as:
\begin{equation}
  y^{(m)} = \frac{l - 1}8 \left((l - 1)\sgn y^{(m)} + 4 ( \Delta^{(m)}
    - p )\right)\,.
\end{equation}
In terms of this new variable \eqref{l(m)} reads
\begin{equation}\label{l(m)new}
  l^{(m)} = l - \sgn y^{(m)} - 1 + 2 \text{Ceil} \left( \Delta^{(m)} +
    \delta_1(\Delta^{(m)}, y^{(m)}) \right)
\end{equation}
with
\begin{equation}\label{delta1}
  \delta_1(\Delta, y) = \begin{cases}
    \frac{1-l}2 + p - \Delta + \frac12 \sqrt{ l^2 + l ( 4 \Delta - 4 p - 2) +
      4 p^2 - 4 \Delta + 2} & y > 0 \\
    \frac{1-l}2 - p + \Delta + \frac12 \sqrt{ l^2 + l ( 4 p - 4 \Delta - 2) +
      4 p^2 + 4 \Delta + 2} & y < 0\,.
  \end{cases}
\end{equation}
In the same fashion, we rewrite \eqref{iteration} as
\begin{equation}
  \Delta^{(m+1)} = \Delta^{(m)} + 1 + 2 p - 2 \text{Ceil} \big( \Delta^{(m)} +
    \delta_1(\Delta^{(m)}, y^{(m)}) \big) + \delta_2(\Delta^{(m)}, y^{(m)})
\end{equation}
with
\begin{equation}\label{delta2}
  \delta_2(\Delta, y) = \sgn y \, \frac{( 1 + 4 p - 2 \text{Ceil} (\Delta))
    (2 \text{Ceil}(\Delta) - 1)}{2(l - 1)}\,.
\end{equation}
We repeat these steps also for the inverse evolution \eqref{iterationinv} and
\eqref{tildel(m)} to obtain
\begin{equation}
  \tilde l^{(m)} = l - \sgn ( y^{(m)} - p ) - 1 + 2 \text{Ceil} \left(
    \Delta^{(m)} - 2 p + \delta_1 ( \Delta^{(m)}, y^{(m)} - p ) \right)
\end{equation}
and
\begin{equation}\label{inviter}
  \Delta^{(m-1)} = \Delta^{(m)} + 1 - 2 p - 2 \text{Ceil} \left( \Delta^{(m)}
    - 2 p + \delta_1 ( \Delta^{(m)}, y^{(m)} - p ) \right) +
    \tilde \delta_2(\Delta^{(m)}, y^{(m)})
\end{equation}
with
\begin{equation}\label{tildedelta2}
  \tilde \delta_2( \Delta^{(m)}, y^{(m)}) =  \sgn ( y - p ) \, \frac{\Big( 1
    - 4 p - 2 \text{Ceil} (\Delta - 2 p)\Big)\Big( 2 \text{Ceil}(\Delta
    - 2 p) - 1 \Big)}{2(l - 1)}\,.
\end{equation}
Note that all the results presented so far are just a mere rewriting of
the equations \eqref{iteration}, \eqref{l(m)}, \eqref{iterationinv} and
\eqref{tildel(m)}. They do not involve any approximations. In this form
all contributions from the square root in \eqref{l(m)} and the non-linear
parts in \eqref{iteration} are captured by the functions
$\delta_1(\Delta, y)$ and $\delta_2(\Delta, y)$. The same holds for the
inverse iteration. Following this observation, we implement a similar
splitting for $\Delta^{(m)}$. More specifically, it splits into a leading
contribution $\Delta^{(m)}_\mathrm{L}$ (L for leading) and corrections
$\delta^{(m)}$
\begin{equation}\label{DeltaSplitting}
  \Delta^{(m)} = \Delta^{(m)}_\mathrm{L} + \delta^{(m)}\,.
\end{equation}
To make this splitting unique, we make two self-consistent definitions:
\begin{equation}
  \delta^{(1)} \equiv - \delta_1(\Delta^{(1)}, y^{(1)})
\end{equation}
and
\begin{equation}\label{leading}
  \Delta^{(m+1)}_\mathrm{L} \equiv \Delta^{(m)}_\mathrm{L} + 1 + 2 p -
    2 \text{Ceil} ( \Delta^{(m)}_\mathrm{L} )\,.
\end{equation}
An immediate consequence of substituting equation \eqref{DeltaSplitting} into
equation \eqref{l(m)new} for $m=1$ is that
\begin{equation}\label{boundDeltaL1}
  -1 < \Delta^{(1)}_\mathrm{L} < 1 \,.
\end{equation}

The major advantage of the leading contribution $\Delta^{(m)}_\mathrm{L}$
is that the only non-linearity it contains is the ceiling function. Thus,
we can derive an analytic expression for the time evolution
$\Delta^{(m)}_\mathrm{L}$ starting from the initial value
$\Delta^{(1)}_\mathrm{L}$. First, we go one step further than before
\begin{equation}
  \Delta^{(m+2)}_\mathrm{L} = \Delta^{(m)}_\mathrm{L} + 4 p +
    2 \text{Ceil} ( \Delta^{(m)}_\mathrm{L} ) - 2 \text{Ceil} (%
    \Delta^{(m)}_\mathrm{L} + 2 p )\,.
\end{equation}
Applying this equation multiple times implies
\begin{equation}\label{solDeltaL}
  \Delta^{(m+2 n)}_\mathrm{L} = \Delta^{(m)}_\mathrm{L} + 4 n p +
    2 \sum\limits_{l=0}^{2n-1} (-1)^l \, \text{Ceil}(%
    \Delta^{(m)}_\mathrm{L} + 2 l p )\,.
\end{equation}

Here then, is the crux of the analysis: Although we started with a highly
non-linear dynamical system with both a square root function and a ceiling
function, the reduced evolution equation as captured by equation
(\ref{solDeltaL}) only involves integral quantities. As such, it is easier to
analyze the periodic behavior of this system. In particular, we can already
anticipate that there is something special about the value of
$b \text{ mod } 4$ in the ratio $p = a / b$ because of the factors
of $4$ appearing in equation (\ref{solDeltaL}). The caveat to this is of course,
that the ``rounding errors'' that we make at each stage of our evolution
equation remain bounded and small. One of the purposes of Appendix \ref{app:DOMINO}
is to establish this property in the special limit where the initial value $x(1)$
is sufficiently large (as measured in units of $b$).

Qualitatively, we also understand from lemma \ref{lemma2} in Appendix
\ref{app:DOMINO} why orbits with $b \text{ mod } 4$ do not close (even if we do
not give a formal proof by analyzing all the rounding errors explicitly).
This lemma states that after $b/2$ extrema, the dominant contribution
$\Delta^{(m)}_\mathrm{L}$ increases by one:
\begin{equation}
  \Delta^{(m+b/2)}_\mathrm{L} = \Delta^{(m)}_\mathrm{L} + 1\,.
\end{equation}
Neglecting the subleading corrections, this results implies that the
length of contiguous blocks of $+/-$ grows over time and therewith
$|x(i_\text{min})|$/$|x(i_\text{max})|$. Figure \ref{fig:nonperiodicorbit}
demonstrates this effect for $p(1)=5/4$ and $x(1)=1/4$.
\begin{figure}[t!]
  \centering
  \includegraphics{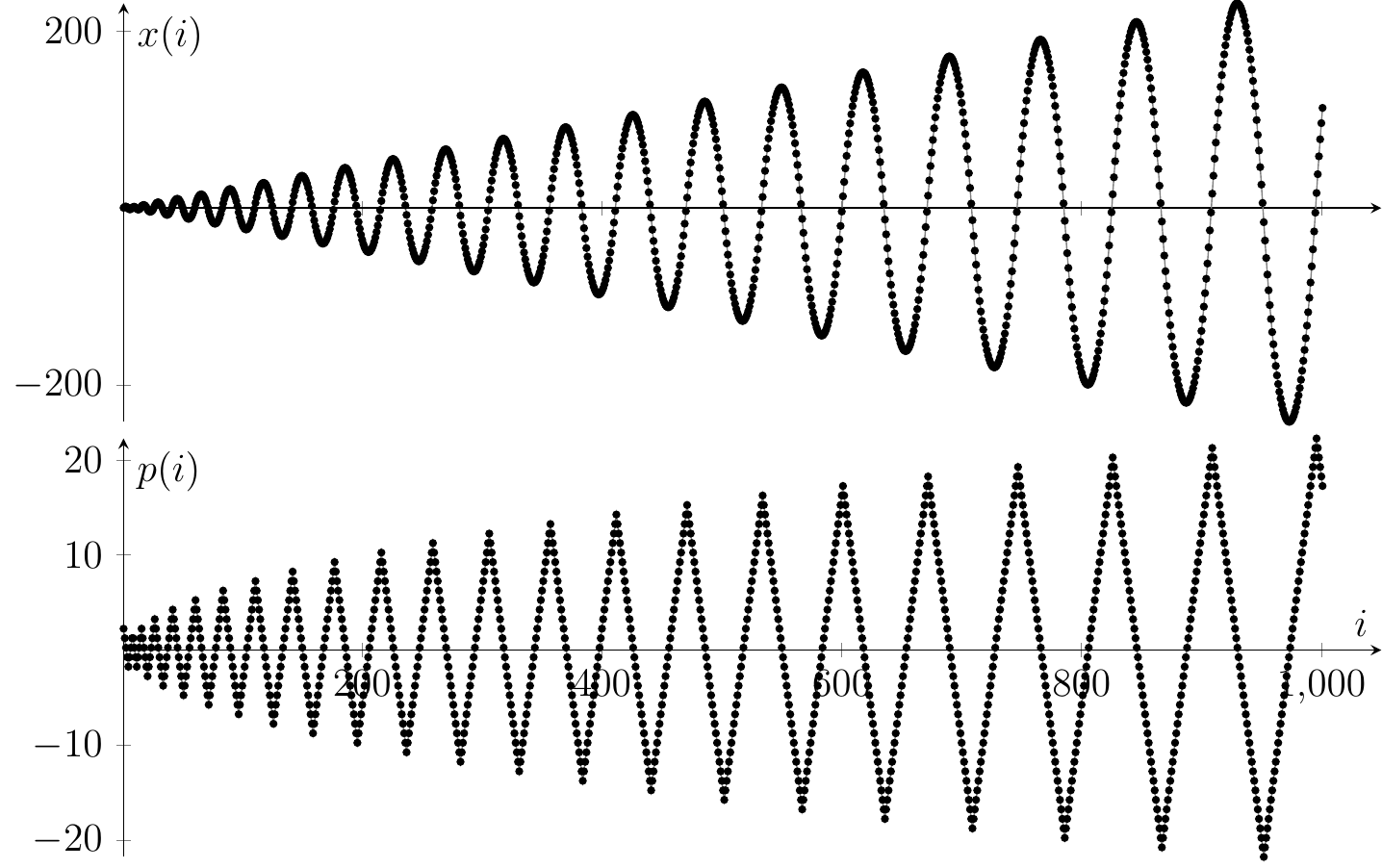}
  \caption{A generic feature of initial conditions with $p=a/b$ and $b\text{ mod }4=0$
is that the enveloping curves of $x(i)$ and $p(i)$ grow over time. Therefore all
examples we have studied numerically do not seem to close into a periodic orbit. Here,
this feature is depicted for the initial conditions  $p(1)=5/4$, $x(1)=1/4$ and the first
1000 $x(i)$/$p(i)$ of the time evolution.}
\label{fig:nonperiodicorbit}
\end{figure}

Another byproduct of our analysis is that when $x(1)$ is sufficiently large and
we have a periodic orbit, we have a ``dominant orbit,'' namely, the leading order
contribution completely captures the evolution of the dynamical system, and all
rounding errors remain small. For such dominant orbits, we can even
obtain an analytic formula for the value of $k$, the length of the orbit:
\begin{equation}
k = l^{(\ast)} k' + 2 k' p
\end{equation}
where $l^{(\ast)}$ is a particular choice of $l^{(m)}$, as defined in equation
(\ref{l(m)}), and $k'$ is the number of extrema in the orbit. The particular value
of $m$ which enters here is the one for which:
\begin{equation}
0 < \mathrm{sgn} (y^{(m)}) \Delta_{\mathrm{L}}^{(m)} < \frac{1}{k'}.
\end{equation}

\section{Terminal Punctures \label{sec:TERMINAL}}

Having understood some of the general properties of eternal punctures, we now
turn to the structure of the dynamical system in the case of terminal
punctures. In this case, it can happen that the full description of
the puncture will decompose into several independent pieces, because at each
termination point, we can restart at the next time step with a new choice of
initial condition.

With this in mind, we consider a puncture for which $x(0)=x(k)=0$. Given some
choice of initial conditions, we would like to track the subsequent evolution
under equation (\ref{evolve}):%
\begin{equation}
\left[
\begin{array}
[c]{c}%
p(i+1)\\
x(i+1)
\end{array}
\right]  =\left[
\begin{array}
[c]{cc}%
1 & 0\\
1 & 1
\end{array}
\right]  \left[
\begin{array}
[c]{c}%
p(i)\\
x(i)
\end{array}
\right]  -\left[
\begin{array}
[c]{r}%
\text{sgn }x(i)\\
\text{sgn }x(i)
\end{array}
\right]  . \label{evolagain}%
\end{equation}
To begin, we have the necessary condition:%
\begin{equation}
x(1)=p(1)=p(0),
\end{equation}
so in some sense, everything is dictated by the choice of a single real
parameter, $p(0)$. Remarkably, the late time behavior of the dynamical system turns out to be
quite sensitive to this initial condtion.

Our next task is to understand which values of this initial condition can
produce a terminal puncture.\ Along these lines, we observe that upon solving
for the $p(i)$ and $x(i)$, we can iterate until for some $i=i_{\ast}^{(1)}$,
$x(i_{\ast}^{(1)})=0$. Note that by assumption that we have a non-trivial
terminal puncture, this occurs for $2\leq i_{1}^{\ast}\leq k$. Since this also
means $x(i)$ is non-zero for these intermediate values, the puncture equations
and the dynamical system lead to the same conditions on the $x(i)$ and $p(i)$.
Iteratively solving for the $x(i)$, we then obtain the necessary conditions
summarized by the matrix equation:%
\begin{equation}
\left[
\begin{array}
[c]{ccccc}%
2 & -1 &  &  & \\
-1 & 2 & -1 &  & \\
& -1 & \ddots & -1 & \\
&  & -1 & 2 & -1\\
&  &  & -1 & 2
\end{array}
\right]  \left[
\begin{array}
[c]{c}%
x(1)\\
x(2)\\
\vdots\\
x(k-2)\\
x(i_{\ast}^{(1)}-1)
\end{array}
\right]  =\left[
\begin{array}
[c]{c}%
\text{sgn }x(1)\\
\text{sgn }x(2)\\
\vdots\\
\text{sgn }x(i_{\ast}^{(1)}-2)\\
\text{sgn }x(i_{\ast}^{(1)}-1)
\end{array}
\right]  ,
\end{equation}
where the $(i_{1}^{\ast}-1)\times(i_{1}^{\ast}-1)$ matrix is the Cartan matrix
for an $A$-type root system, in the obvious notation. Inverting the Cartan
matrix, we see that all of the $x(i)$ are necessarily rational numbers, so we
must require our initial condition $p(0)$ to be a rational number, which we
write as:%
\begin{equation}
p(0) = p(1)=x(1)=\frac{a(1)}{b}\text{ \ \ for }a(1),b\text{ relatively prime.}%
\end{equation}

In the context of the dynamical system, nothing stops us from continuing to
evolve our values of the $x(i)$. As we have already remarked, a
\textquotedblleft short orbit\textquotedblright\ is one for which $x(i_{\ast
}^{(1)}+1)=x(1)$, whereas for a \textquotedblleft long orbit\textquotedblright%
\ we must instead cycle around once more before we repeat back to $x(1)$. In
the context of finding punctures, of course, there is no need to continue
evolving with the dynamical system once we reach a termination point. Indeed,
we are free to restart the dynamical system with some other choice of initial
conditions, in which case we obtain a terminal puncture with multiple
irreducible components.

We observe that we can pass from a terminal puncture back to an eternal
puncture by making a sufficiently small perturbation in the value of $x(1)$,
namely:%
\begin{equation}
x(1)\mapsto x(1)+\delta x.
\end{equation}
If the terminal puncture yields a periodic orbit, it follows that the system
with the perturbed initial condition will also yield a periodic orbit. Indeed,
this is the key feature of having a small continuous parameter in the eternal punctures.

But for the eternal punctures we already argued that to reach a periodic
orbit, a necessary and sufficient condition is that $p(1)=a(1)/b$ with $b\neq0$
mod $4$. So, we see that this condition also hold for terminal punctures:%
\begin{equation}
\text{Terminal Puncture:\ } p(1)=x(1)=\frac{a(1)}{b}\text{ \ \ with }b\text{
}\neq\text{mod }4\text{.}%
\end{equation}
We reach these special orbits by appropriate tuning of the parameter $\delta
x$.

Now, from the perspective of the puncture equations, we can equally well
characterize the evolution by the values of $x(i)$, or by simply stating the
sequence of signs that are obtained from a particular choice of initial
condition. As we change the value of the initial condition, however, this
sequence will bifurcate to two distinct sequences, and can then bifurcate
further after some further number of time steps. Of course, the delicate part
in the analysis is that a priori, we do not know which of the $2^{k}$ possible
sequences will actually yield a solution to the puncture equations.

Let us give a few examples to illustrate the general idea. For $p(0) = p(1) = k/2$ we find
that $\mathrm{sgn}(x(i))=+1$ for $1\leq i\leq k$. Moreover, the parent matrix
$\widetilde{Q}=0$, so we generate a representation of $\mathfrak{su}(2)$. Indeed, the
eigenvalues of the generator of the Cartan subalgebra in such a representation
are always half integers or integers, so this is not
particularly surprising. Note also that if $p(1) =a(1) / b$ has $b$ different from $1$
or $2$, we necessarily do not have a representation of $\mathfrak{su}(2)$.

A non-trivial class of examples in which the signs alternate occurs when
$k=2n$ is even:%
\begin{equation}
l=\frac{n}{2n+1}\Rightarrow\left\{  \text{sgn }x(i)\right\}  _{i=1}%
^{k}=\underset{2n}{\underbrace{+,-,...,+,-}.}%
\end{equation}

To determine viable sequences of signs, we therefore use a standard tool in
dynamical systems: We consider without loss of generality the value of the
parameter $l$, and seek out bifurcation points where the sign of $x(i)$ is
about to split. These special bifurcation points tell us the locations of
finite length periodic orbits. Iterating in this way, we can determine both
the initial condition as well as the corresponding value of $k$ appearing in
the puncture equations.

Applying the evolution equation (\ref{evolagain}) on the interval $[0,\infty
)$, this line segment splits into two segments specified by the conditions
$x(2)<0$ and $x(2)>0$. The bifurcation point amounts to splitting the interval
up into two subintervals for possible values of the initial condition, namely
$0\leq p(0) <1/2$ and $1/2< p(0) <\infty$. By repeating this procedure one finds more
and more new line segments and the partition of the parameter interval for $p(0)$
gets finer. Figure~\ref{fig:lineseg} depicts this process for the first three
iterations.
\begin{figure}[t]
  \centering
  \includegraphics{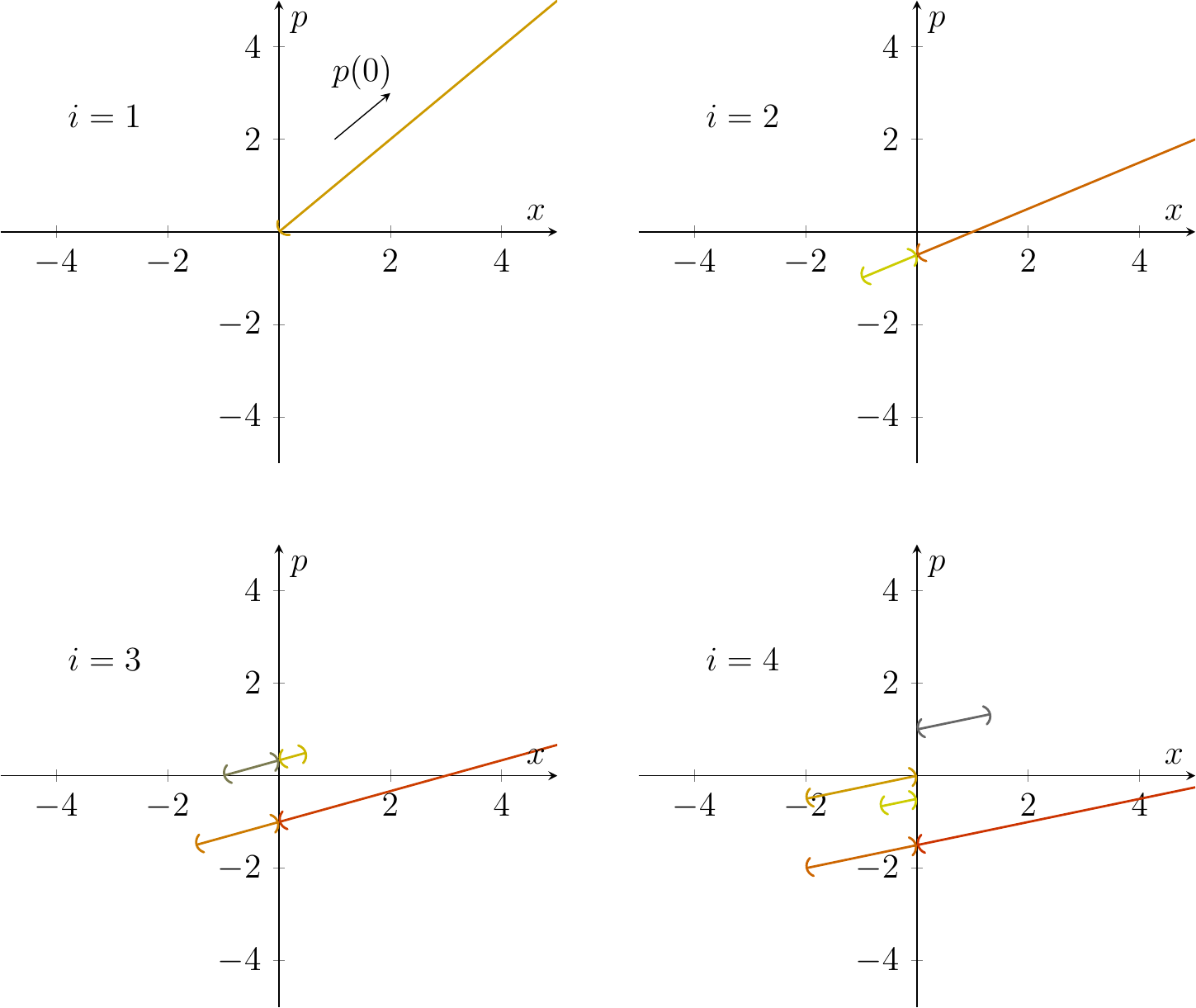}
  \caption{Three iterations of the initial line segment $[0 , \infty )$ which give rise to bifurcating behavior
  in the initial condition $l$.}
  \label{fig:lineseg}
\end{figure}
In the upper left corner, there is the initial line segment
$[0,\infty)$ which is completely in the right quadrant with $x>0$. In order to
iterate this segment, we apply the evolution equation (\ref{evolagain}) to
each of its points and so obtain the diagram for $i=2$ on the upper right.
Moreover, we split the line segment into two parts so that each one only has
points in one quadrant either $x>0$ or $x<0$. After iterating and splitting
them again, one arrives at the diagram for $i=3$ in the lower left corner. So
far the number of line segments has grown exponentially as a function of $k$.
This drastically changes starting with the fourth iteration, resulting in the
last diagram of figure~\ref{fig:lineseg}. Here only one of the four line
segments arising by iteration from the last step is located in two quadrants
and therefore needs splitting. Of course we can also identify the various
segments in this diagram with the following intervals:
\begin{equation}
\hspace*{-1em}%
\begin{tikzpicture} \matrix (m) [matrix of math nodes, column sep=4em] { i=1 & 0 & & & & & \infty\,. \\[1em] i=2 & 0 & & 1/2 & & & \infty\,. \\[1em] i=3 & 0 & 1/3 & 1/2 & 1 & & \infty\,. \\[1em] i=4 & 0 & 1/3 & 1/2 & 1 & 3/2 & \infty\,. \\ k & 1 & 3 & 2 & 3 & 4 & \\ }; \draw[lineC=500,(-,thick] (m-1-2.east) -- (m-1-7.west |- m-1-2) node[midway, above,black] {\tt +}; \draw[lineC=333,(-),thick] (m-2-2.east) -- (m-2-4.west |- m-2-2) node[midway, above,black] {\tt +-}; \draw[lineC=666,(-,thick] (m-2-4.east) -- (m-2-7.west |- m-2-4) node[midway, above,black] {\tt ++}; \draw[lineC=200,(-),thick] (m-3-2.east) -- (m-3-3.west |- m-3-2) node[midway, above,black] {\tt +--}; \draw[lineC=400,(-),thick] (m-3-3.east) -- (m-3-4.west |- m-3-3) node[midway, above,black] {\tt +-+}; \draw[lineC=600,(-),thick] (m-3-4.east) -- (m-3-5.west |- m-3-4) node[midway, above,black] {\tt ++-}; \draw[lineC=800,(-,thick] (m-3-5.east) -- (m-3-7.west |- m-3-5) node[midway, above,black] {\tt +++}; \draw[lineC=167,(-),thick] (m-4-2.east) -- (m-4-3.west |- m-4-2) node[midway, above,black] {\tt +--+}; \draw[lineC=333,(-),thick] (m-4-3.east) -- (m-4-4.west |- m-4-3) node[midway, above,black] {\tt +-+-}; \draw[lineC=500,(-),thick] (m-4-4.east) -- (m-4-5.west |- m-4-4) node[midway, above,black] {\tt ++--}; \draw[lineC=667,(-),thick] (m-4-5.east) -- (m-4-6.west |- m-4-5) node[midway, above,black] {\tt +++-}; \draw[lineC=833,(-,thick] (m-4-6.east) -- (m-4-7.west |- m-4-6) node[midway, above,black] {\tt ++++}; \end{tikzpicture}
\end{equation}
of the parameter $p(1)$ where we use the same color coding as in
figure~\ref{fig:lineseg}. Additionally, the signs $+$ and $-$ denote the
quadrants ($x(i)>0$ and $x(i)<0$) which these segments are in after the $i$-th
iterations. They are read from left to right. It is obvious that the first one
is always $+$, because the initial interval is completely in the quadrant
$x(1)>0$. Doing further iterations, result in a refinement of this diagram.
There are a few initial conditions $l\in\{1/3,1/2,1,3/2\}$ which we excluded
so far. They give rise to the four smallest periodic orbits of the dynamical
system. Their length $k$ is equivalent to the value of $i$ where the value of
$l$ appears as a interval boundary for the first time. For a depiction of the
branching, see figure \ref{fig:orbittree}.
\begin{figure}
  \centering
  \includegraphics{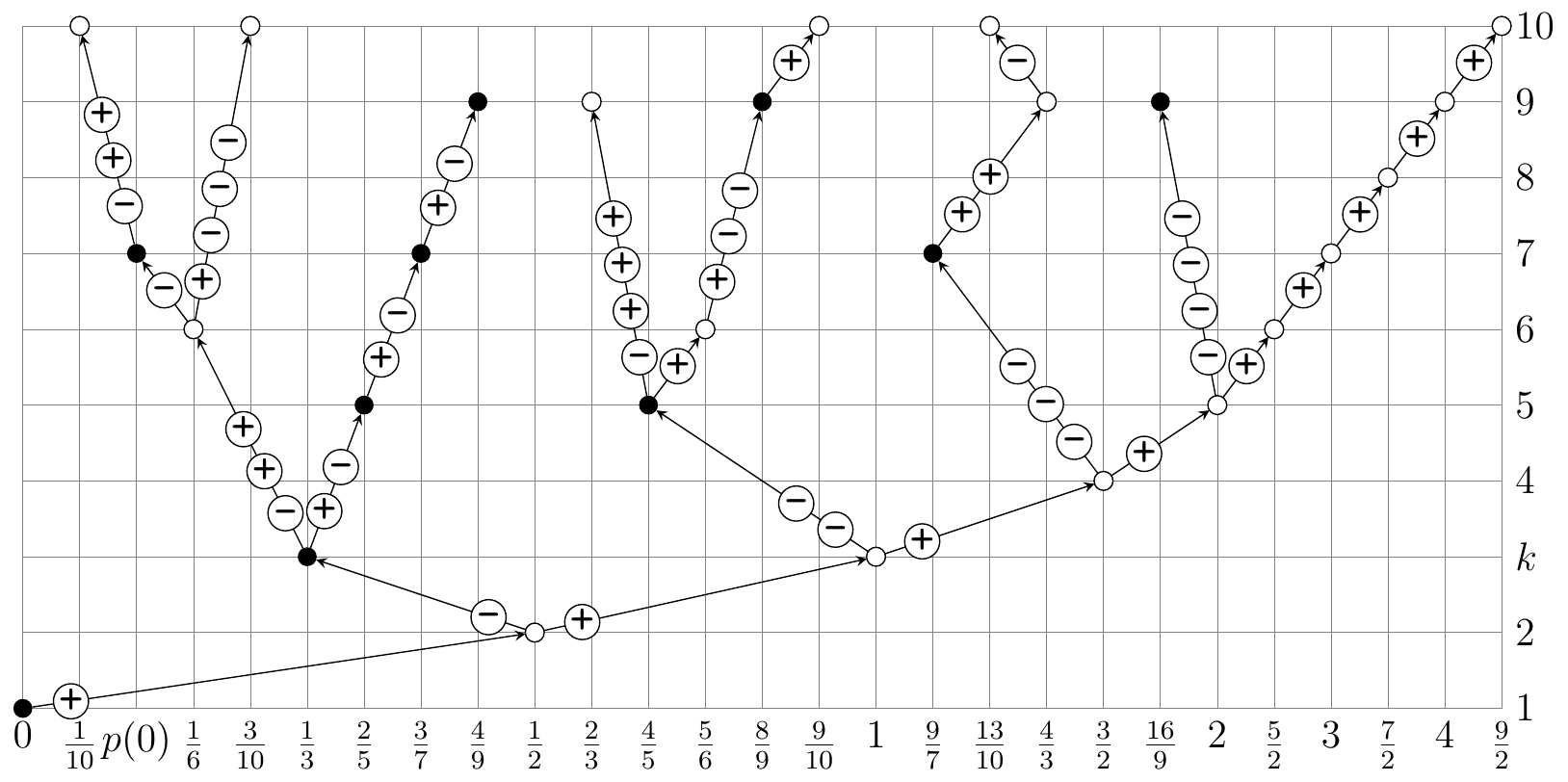}
  \caption{Periodic orbits which represent the 1/2 BPS punctures of the $A_k$ theory for $k\le 10$ if one link between the quiver nodes is turned off. Whether a node is filled or not distinguishes between short and long
orbits. Each link between two nodes carries the sequence of $\sgn x(i)$ which has to be appended to the shorter orbit in order to obtain the longer one.}\label{fig:orbittree}
\end{figure}

\subsection{Irreducible Representations}

Though the dynamical system is clearly quite sensitive to initial condition
data, it nevertheless exhibits some regular features, at least for certain
values of the orbit length $k$. To track this behavior, we first list out
$m_{k}$, the number of terminal punctures with length $k$ and $x(k)=0$, with
no vanishing of any $x(i)$ for $0<i<k$. The multiplicity of these
\textquotedblleft irreducible representations\textquotedblright\ for small
values of $k$ are as follows:%
\begin{equation}%
\begin{tabular}
[c]{cccccccccccccccccccc}%
$k$ & 2 & 3 & 4 & 5 & 6 & 7 & 8 & 9 & 10 & 11 & 12 & 13 & 14 & 15 & 16 & 17 &
18 & 19 & 20\\
$m_{k}$ & 2 & 4 & 2 & 6 & 6 & 8 & 2 & 12 & 10 & 12 & 6 & 14 & 14 & 22 & 2 &
18 & 18 & 20 & 10\thinspace.
\end{tabular}
\end{equation}
This sequence contains some regular subsequences as can already be seen from
the first contributions but becomes more obvious by looking at a larger
number of $k$ as depicted in figure~\ref{fig:nosol}.

Some of the multiplicities satisfy rather regular patterns. For example, we
have:
\begin{equation}
m_{k_{j}(n)}=4n-2\quad\text{for all}\quad j\in\mathbb{N}\quad\text{with}\quad
k_{j}(n)=2^{j}(2n-1)\,.
\end{equation}
All solutions belonging to these subsequences are long orbits. In total this
regular contribution fixes half of the sequence $m_{k}$, namely for $k$ even.
This also explains why it only contains long orbits. Short orbits require $k$
to be odd. For odd $k$, there seems to be no regular pattern, as can also be
seen from figure \ref{fig:nosol}.

\begin{figure}[t]
  \centering
  \includegraphics{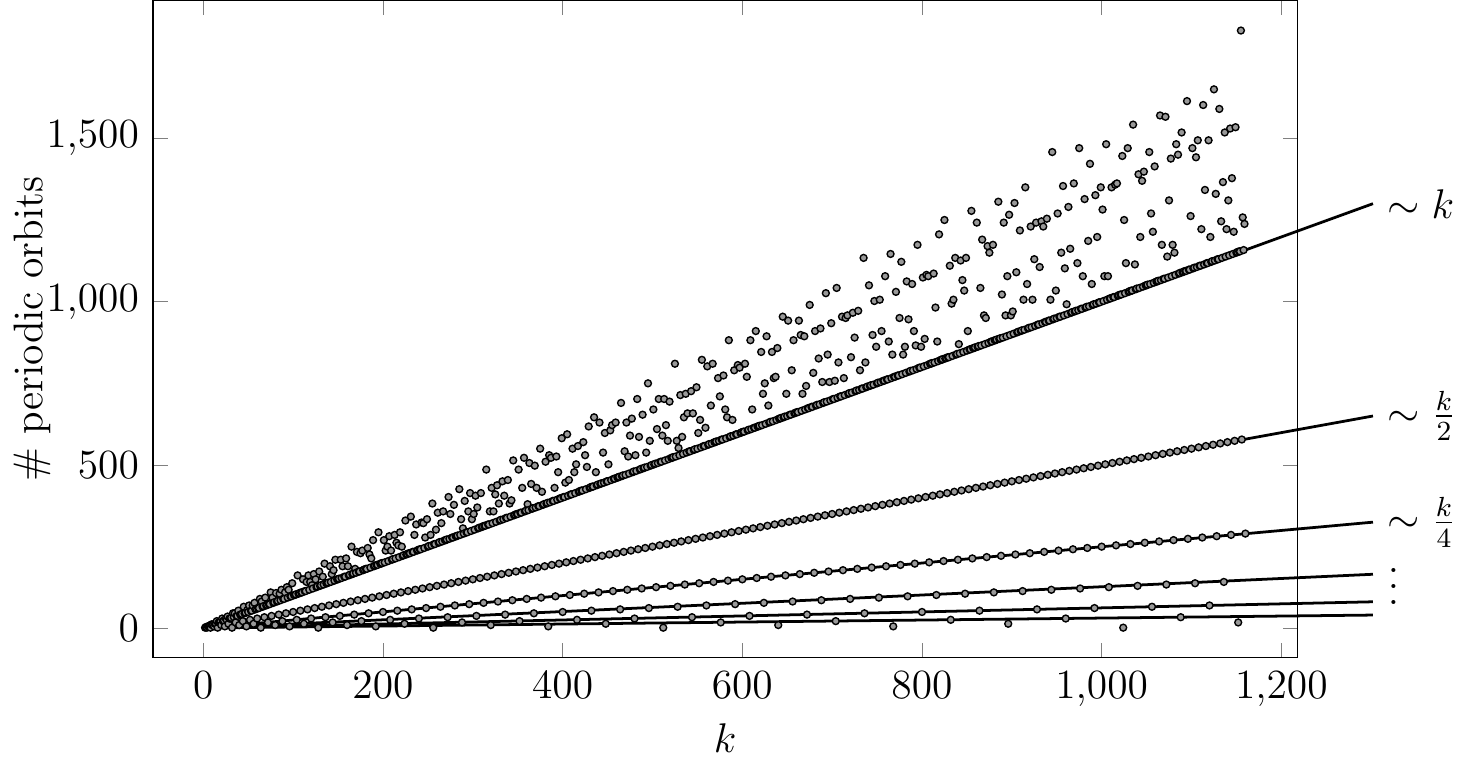}
  \caption{Number of periodic orbits with the length $k$. The regular
subsequences are marked by lines with slope $k/2^{j}$.}%
\label{fig:nosol}%
\end{figure}

\section{Higher Order Poles \label{sec:POLES}}

Having presented a general characterization of simple poles, we now turn to the structure of punctures
with higher order poles. As we have already remarked in section \ref{sec:DYNAMO},
the structure of the first order pole involves non-linear quadratic evolution in
the values of the $q$ and $\widetilde{q}$'s. At second order and above,
however, the structure becomes linear in these parameters, since the lower
order poles essentially serve as \textquotedblleft boundary
conditions\textquotedblright\ for evolution in the pole order. It thus follows
that as we evolve from one pole order to the next, there is sometimes a degree
of freedom associated with whether we continue to higher order, but for the
most part, all of this is fully fixed by the lower order terms. This holds for
general values of $k$ and $N$ in our puncture equations. Here, we only study the $N=1$
case explicitly for which we have full control over the first order. With this
in mind, we assume that we have been given a consistent solution $q_1(i)$, $\tilde q_1(i)$
and $p_1(i)$, and then deduce the structure of the higher order poles for this puncture.
Indeed, for ``generic'' choices of initial conditions at lower order, we can evolve to a
higher order solution. There are, however, possible obstructions to continuing this evolution indefinitely to higher order poles,
and part of our aim in this section will be to determine these obstructions. Modulo this caveat, however, we see a sharp
sense in which the dynamical system involves evolution in ``time'' (namely from one quiver node to its neighbor) and ``space'' (namely from one pole order to the next).

For the higher order contributions $\Sigma_m(i)$, $Q_m(i)$ and $\tilde Q_m(i)$, we use
the same parameterization as for the leading ones in \eqref{SIGMA1}-\eqref{Qtilde1}.
It is convenient to combine this data to a vector
\begin{equation}
\vec{v}_{m}(i)=%
\begin{bmatrix}
q_{m}(i) & \tilde{q}_{m}(i) & p_{m}(i)
\end{bmatrix}
\end{equation}
for every node of the $A$-type quiver. The constraint equations
\eqref{higherpoleconstr} gives rise to a system of affine maps
system of affine maps which connect the puncture data along the quiver.
We split this system into a linear part $M_{m}(i)$ plus the
offset contribution $\vec{n}_{m}(i)$. It gives rise to the iteration
prescription
\begin{equation}
\vec{v}_{m}(i+1)=M_{m}(i)\vec{v}_{m}(i)+\vec{n}_{m}(i)
\label{eqn:iterationvecvm}%
\end{equation}
for the puncture data at level $m$. Assuming that we can invert and solve the
first order pole equations, the explicit expressions for $M_{m}(i)$ and
$\vec{n}_{m}(i)$ follow after some algebra directly from \eqref{higherpoleconstr}
\begin{equation}
M_{m}(i)=%
\begin{cases}%
\begin{bmatrix}
-\alpha_{m}(i)+\beta_{m}(i) & 0 & \phantom{-}\delta_{m}(i)\\
0 & \beta_{m}(i) & 0\\
-\gamma_{m} & 0 & 1
\end{bmatrix}
& \text{for }x(i)>0\text{ and }x(i+1)>0\\%
\begin{bmatrix}
0 & \beta_{m}(i) & 0\\
\phantom{-}\alpha_{m}(i)-\beta_{m}(i) & 0 & -\delta_{m}(i)\\
-\gamma_{m}(i) & 0 & 1
\end{bmatrix}
& \text{for }x(i)>0\text{ and }x(i+1)<0\\%
\begin{bmatrix}
0 & \phantom{-}\alpha_{m}(i)-\beta_{m}(i) & \phantom{-}\delta_{m}(i)\\
\beta_{m}(i) & 0 & 0\\
0 & \gamma_{m}(i) & 1
\end{bmatrix}
& \text{for }x(i)<0\text{ and }x(i+1)>0\\%
\begin{bmatrix}
\beta_{m}(i) & 0 & 0\\
0 & -\alpha_{m}(i)+\beta_{m}(i) & -\delta_{m}(i)\\
0 & \gamma_{m}(i) & 1
\end{bmatrix}
& \text{for }x(i)<0\text{ and }x(i+1)<0
\end{cases}
\end{equation}
and
\begin{equation}
\vec{n}_{m}(i)=%
\begin{cases}%
\begin{bmatrix}
\phantom{-}\gamma_{m}^{\prime}(i)+\alpha_{m}^{\prime}(i)\delta_{m}(i) &
\beta_{m}^{\prime}(i) & \phantom{-}\alpha_{m}^{\prime}(i)
\end{bmatrix}
^{T} & \text{for }x(i)>0\text{ and }x(i+1)>0\\%
\begin{bmatrix}
\beta_{m}^{\prime}(i) & -\gamma_{m}^{\prime}(i)-\alpha_{m}^{\prime}%
(i)\delta_{m}(i) & \phantom{-}\alpha_{m}^{\prime}(i)
\end{bmatrix}
^{T} & \text{for }x(i)>0\text{ and }x(i+1)<0\\%
\begin{bmatrix}
\phantom{-}\gamma_{m}^{\prime}(i)-\tilde{\alpha}_{m}^{\prime}(i)\delta
_{m}(i) & \beta_{m}^{\prime}(i) & -\tilde{\alpha}_{m}^{\prime}(i)
\end{bmatrix}
^{T} & \text{for }x(i)>0\text{ and }x(i+1)<0\\%
\begin{bmatrix}
\beta_{m}^{\prime}(i) & -\gamma_{m}^{\prime}(i)+\tilde{\alpha}_{m}^{\prime
}(i)\delta_{m}(i) & -\tilde{\alpha}_{m}^{\prime}(i)
\end{bmatrix}
^{T} & \text{for }x(i)>0\text{ and }x(i+1)<0
\end{cases}
\end{equation}
In order to write them in a compact form, the abbreviations
\[
\alpha_{m}(i)=\gamma_{m}(i)\delta_{m}(i)\,,\quad\beta_{m}(i)=\sqrt{\left\vert
\frac{x(i)}{x(i+1)}\right\vert }\,,\quad\gamma_{m}(i)=\frac{m-2}{\sqrt
{|x(i)|}}\,,\quad\delta_{m}(i)=\frac{m-1}{2\sqrt{|x(i+1)|}}%
\]
and
\[
\alpha_{m}^{\prime}(i)=\frac{a_{m}(i)}{\sqrt{x(i)}}\,,\quad\tilde{\alpha}%
_{m}^{\prime}(i)=\frac{\tilde{a}_{m}(i)}{\sqrt{x(i)}}\,,\quad\beta_{m}%
^{\prime}(i)=-\frac{b_{m}(i)}{\sqrt{|x(i+1)|}}\,,\quad\gamma_{m}^{\prime
}(i)=\frac{c_{m}(i+1)}{2\sqrt{|x(i+1)|}}%
\]
are convenient. For the next to leading order contribution $m$=$3$, the offset
part $\vec{n}_{m}(i)$=$0$ vanishes and we only have to take into account
$M_{m}(i)$. Once an initial value $\vec{v}_{n}(1)$ is chosen,
\eqref{eqn:iterationvecvm} allows to calculate all remaining $\vec{v}_{n}(i)$
iteratively. Hence, we only have to specify an appropriate initial value. To
this end, we employ the monodromy
\begin{equation}
\mathbf{M}_{m}=M_{m}(k)M_{m}(k-1)\dots M_{m}(1)
\end{equation}
which has to fulfill
\begin{equation}
\vec{v}_{3}(k+1)=\mathbf{M}_{3}\vec{v}_{3}(1)=\vec{v}_{3}(1)
\end{equation}
to be compatible with the periodic boundary conditions of the quiver. Note
that $\mathbf{M}_{m}$ is an element of $SL(3,\mathbb{C})$. This property is
based on the determinant of $M_{m}(i)$
\begin{equation}
\det M_{m}(i)=\beta_{m}(i)^{2}=\frac{|x(i)|}{|x(i+1)|}%
\end{equation}
which gives immediately rise to
\begin{equation}
\det\mathbf{M}_{m}=\prod\limits_{i=1}^{k}\frac{|x(i)|}{|x(i+1)|}=\frac
{|x(1)|}{|x(k+1)|}=1\,.
\end{equation}
However, not every eigenvector of $\mathbf{M}_{3}$ with eigenvalue $1$ is a
valid initial condition. There are is the additional constraint
\begin{equation}
\tilde{q}_{m-1}(i)=\frac{\tilde{a}_{m}(i)}{m}\quad(x(1)>0)\quad\quad
\text{or}\quad q_{m-1}(i)=\frac{a_{m}(i)}{m}\quad(x(1)<0)
\label{eqn:constrinitcond}%
\end{equation}
depending whether $x(1)$ is positive or negative. It also follows from
\eqref{higherpoleconstr} and is preserved under the iteration prescription
\eqref{eqn:iterationvecvm}. Thus, the non-trivial part of the monodromy is
reduced from $SL(3,\mathbb{C})$ to $SL(2,\mathbb{C})$ and a fixed point which
respects \eqref{eqn:constrinitcond} can only arise if Tr$\mathbf{M}_{3}=3$ holds.
This is because one eigenvalue of $\mathbf{M}_{3}$ has to be $1$ due to
\eqref{eqn:constrinitcond}. Furthermore, the two remaining eigenvalue,
$\lambda_{1}$ and $\lambda_{2}$ are either complex or real. If they are
complex, $\lambda_{2}=\lambda_{1}^{\dagger}$ has to hold because
$\mathbf{M}_{3}$ is real. Additionally, we find
\begin{equation}
\lambda_{1}\lambda_{2}=1\quad\text{and therefore}\quad\lambda_{1}=\frac
{1}{\lambda_{2}}%
\end{equation}
as it is required for $\det\mathbf{M}_{3}=1$. We cannot directly use the
complex eigenvalues, because their eigenvalues are necessary complex too but
$\vec{v}_{3}(1)$ has to be real. Still, we are able to assign $\vec{v}_{3}(1)$
to the sum of the two eigenvectors. Like their eigenvalues, they are complex
conjugated to each other and so their sum is real. In order to obtain a fix
point this ansatz requires $\lambda_{1}+\lambda_{2}=2$ which implies
$\tr\mathbf{M}_{3}=3$. In this case there is exactly one fixed point. It is
unique up to a rescaling. For real eigenvalues at least $\lambda_{1}$ or
$\lambda_{2}$ has to be one. Otherwise there is no fixed point. But in this
case all eigenvalues are one and $\mathbf{M}_{3}$ is the identity matrix
$1_{3}$. As a consequence $\vec{v}_{3}(1)$ is not restricted at all. Rescaling
of $\vec{v}_{3}(1)$ has a very natural interpretation in terms of constraints
\eqref{higherpoleconstr}. They are invariant under the transformation
\begin{equation}
Q_{m}\rightarrow\lambda^{m-1}Q_{m}\,,\quad\tilde{Q}_{m}\rightarrow
\lambda^{m-1}Q_{m}\quad\text{and}\quad\Sigma_{m}\rightarrow\lambda^{m-1}%
\Sigma_{m}%
\end{equation}
with some $\lambda\in\mathbb{R}_{+}$ which is equivalent to $\vec{v}%
_{3}(1)\rightarrow\lambda\vec{v}_{3}(1)$. So without loss of generality, we
fix $\lVert\vec{v}_{3}(1)\rVert=\sqrt{ q_3^2(1) + \tilde q_3^2(1) + p_3^2(1) }=1$.

If we go beyond $m$=3, there is also an offset contribution
\begin{equation}
\vec{\mathbf{n}}_{m}=\vec{n}_{m}(k)+M_{m}(k)\vec{n}_{m}(k-1)+\dots
+M_{m}(k)\dots M_{m}(2)\vec{n}_{m}(1)
\end{equation}
to the fixed point equation
\begin{equation}
\mathbf{M}_{m}\vec{v}_{m}(1)+\vec{\mathbf{n}}_{m}=v_{m}(1)\,.
\label{eqn:fixedpointwithn}%
\end{equation}
By taking into account the constraints on the initial conditions
\eqref{eqn:constrinitcond} this equation is again reduces to a two dimensional
subspace. Let us assume that $\vec{\mathbf{n}}_{m}$ does not vanish. Otherwise
the discussion for the $m=3$ case applies without any modification. We already
have noticed that for $\tr\mathbf{M}_{m}=3$, the matrix $\mathbf{M}_{m}-1_{3}$
does not have full rank. Most severe is this situation if $\mathbf{M}%
_{m}=1_{3}$. In this case \eqref{eqn:fixedpointwithn} does not admit a
solution at all. Alternatively, $\vec{\mathbf{n}}_{m}$ could be in the one
dimensional image of $\mathbf{M}_{m}-1_{3}$. Then, there is a one dimensional
family of fixed points. It is well-known from linear algebra that this family
is the superposition of an inhomogeneous part and all homogeneous
contributions. Latter also give rise to the dimension of the solution space.
They live in the kernel of $\mathbf{M}_{m}-1_{3}$ which is one dimensional for
$\tr\mathbf{M}_{m}=3$ and $\mathbf{M}_{m}\neq1_{3}$ after removing the
direction constrained by \eqref{eqn:constrinitcond}. Finally if $\tr\mathbf{M}%
_{m}\neq3$, $\mathbf{M}_{m}-1_{3}$ is invertible and gives rise to a unique
solution. All these different cases can be summarized in the decision
tree\tikzset{
treenode/.style = {shape=rectangle, rounded corners,
draw, align=center},
env/.style      = {treenode, font=\ttfamily\normalsize},
dummy/.style    = {circle,draw}
}
\begin{equation}
\begin{tikzpicture}[grow=right,level distance=9em, level 1/.style={sibling distance=6em}, level 2/.style={sibling distance=3em}, edge from parent/.style={draw, -latex}, every node/.style={font=\footnotesize},sloped] \node [env] {$\vec{\mathbf{n}}_m$=0} child { node [env] {$\tr \mathbf{M}_m=3$} child { node[dummy] {0} edge from parent node[below] {no} } child { node[env] {$\mathbf{M}_m=1_3$} child { node[env] {$\vec{\mathbf{n}}_m\in\mathrm{Im}(\mathbf{M}_m-1_3)$} child{ node [dummy] {x} edge from parent node[below] {no} } child{ node [dummy] {1} edge from parent node[above] {yes} } edge from parent node[below] {no} } child { node[dummy] {x} edge from parent node[above] {yes} } edge from parent node[above] {yes} } edge from parent node [below] {no} } child { node [env] {$\tr \mathbf{M}_m=3$} child { node [dummy] {x} edge from parent node [below] {no} } child { node [env] {$\mathbf{M}_m=1_3$} child { node[dummy] {1} edge from parent node[below] {no} } child { node[dummy] {2} edge from parent node[above] {yes} } edge from parent node [above] {yes} } edge from parent node [above] {yes} }; \end{tikzpicture}
\end{equation}
Here the number in the circle denotes the dimension of the solution space,
while $x$ represents no solution at all.

In the generic case, $\tr\mathbf{M}_3=3$, $\mathbf{M}_3\ne1_3$ and tr$\mathbf{M}_m\ne 3$,
for all other $m>3$. Then, the solution of the first order problem fixes all higher
orders completely. Note that it is always possible to set all $Q_m(i)$, $\tilde Q_m(i)$
and $\Sigma_m(i)$ to zero for all $m$ larger or equal to $n$. We call such solutions
$n$-trivial. For the generic case there exists $3$-trivial solution, a $4$-trivial
solution and it goes on towards an $\infty$-trivial solution. In the special case, where
one hits the node x for a fixed $m$, there are only $3$-, $\dots$ $m$-trivial solutions.

Finally, there are some special points  where additional degrees of
freedom occur. Identifying them requires to evaluate the trace of the
monodromy $\mathbf{M}_{m}$ for a family of periodic orbits with $x(k)\neq0$.
It is instructive to study a simple example to present the relevant steps.
Take the initial conditions $4/5\leq x(1)=x\leq1$ and $p(1)=3/2$. They give
rise to a periodic orbits of length $k=6$ and
\begin{equation}%
\begin{tabular}
[c]{l|cccccc}%
$i$ & 1 & 2 & 3 & 4 & 5 & 6\\\hline
$x(i)$ & $x$ & $x+1/2$ & $x$ & $x-3/2$ & $x-2$ & $x-3/2$\\
$p(i)$ & $3/2$ & $1/2$ & $-1/2$ & $-3/2$ & $-1/2$ & $1/2$\thinspace.
\end{tabular}
\end{equation}
Now, we calculate
\begin{equation}
\tr\mathbf{M}_{3}-3=\sum_{n=0}^{12}c_{n}(x)m^{n}%
\end{equation}
which is a polynomial in $m$ with coefficients depending on $x$. There is no
simply analytic expression for these coefficients. Still we can calculate them
numerically without much effort. The same is true for the zeros $m_{i}$ of the
polynomial. For every value of $x$ in the given interval they include $3$.
Hence we can always construct a non-trivial next to leading order. Moreover,
there is a distinguished value for $x$ in the interval where $4$ this another
integer zero. A numerical analysis shows that this value is $x=0.88055824$. It
gives rise to a family of $m=4$ solutions with one free parameter.

\section{Comments on the Higher Rank Case \label{sec:HIGHER}}

Much of our focus up to this point has been on the relatively
\textquotedblleft simple\textquotedblright\ case of $N=1$, but arbitrary $k$.
From the perspective of 6D SCFTs, these punctures are most directly associated with
the theory of a single M5-brane probing the transverse geometry $\mathbb{R}%
_{\bot}\times\mathbb{C}^{2}/\mathbb{Z}_{k}$. Of course, given our full
characterization of the $N=1$ and general $k$ case, we can now produce
additional solutions to the puncture equations for $N$ M5-branes probing a
general ADE\ singularity, as well as new $1/4$ BPS\ punctures
of class $\mathcal{S}$ theories. For the class $\mathcal{S}_{\Gamma}$
theories, we follow a similar procedure to that given in reference \cite{Heckman:2016xdl},
namely we pick independent one-dimensional linear subspaces
associated with each quiver node, and draw independent paths which connect to
each such subspace. In reference \cite{Heckman:2016xdl}, this was used to produce
$\mathfrak{su}(2)^{\mathcal{P}}$ representations, where $\mathcal{P}$
denotes the number of independent paths.
In that context, the sign of all of the $x(i)$ of a given path are
fixed to all be the same, namely we obtain a directed path. From the present
perspective, we can now alter the relative signs of these segments to obtain a
far broader class of solutions. In particular, we are guaranteed that these
are new solutions, since as we have already remarked, each such path cannot
correspond to a representation of $\mathfrak{su}(2)$. See figure~\ref{fig:Dtype}
\begin{figure}
  \centering
  \includegraphics{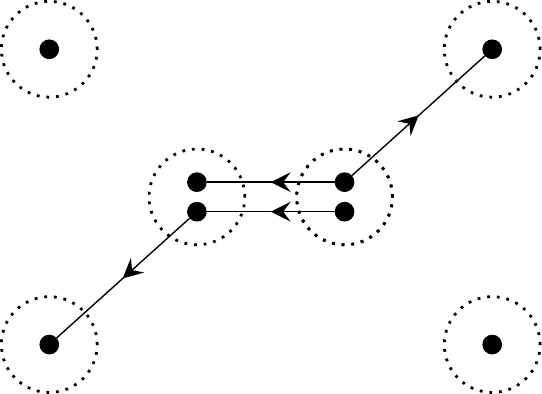}
  \caption{Demonstration of the undirected path approach for a simple D-type quiver.}
  \label{fig:Dtype}
\end{figure}
for a depiction of one such undirected path for a D-type affine quiver.

It is also natural to consider particular small values of $k$, such as $k = 2$, but 
with $N$ arbitrary. Solving the puncture equations in this case seems to be a problem 
of similar difficulty to simply specifying all $1/4$ BPS punctures of class $\mathcal{S}$ 
theories. With this in mind, it seems more fruitful to see what we can say about the 
general structure of solutions as a function of $N$ and $k$.

Indeed, more ambitiously, one might also hope to recast the general rank $N$ case in
terms of a perhaps more involved dynamical system now involving matrices. Our
plan in the remainder of this section will be to show how this can at least be
formally stated, though we defer the corresponding analysis of consistent
initial conditions to future work.

To this end, we first repackage the puncture equations for
the class $\mathcal{S}_{k}$ theories in terms of a recursion relation.
Suppose, then, that we have specific values of $q(i)$, $\widetilde{q}(i)$ and
$p(i)$, which respectively denote position and momenta in the system. In the
generic case, the doublet:%
\begin{equation}
\left[
\begin{array}
[c]{c}%
q(i)\\
\widetilde{q}^{\dag}(i)
\end{array}
\right]
\end{equation}
is a rank $N$ matrix. In fact, we can see that the generic linear combination
of the form:%
\begin{equation}
\widehat{q}_{\varepsilon}(i)\equiv q(i)+\varepsilon\widetilde{q}^{\dag}(i)
\end{equation}
is invertible. One way to establish this is to work in a basis where all
$p(i)$ are diagonal. In this basis, we observe that the matrix entries for the
puncture equations include the constraints:%
\begin{align}
\left[  p_{A}(i)-p_{B}(i+1)\right]  q_{AB}(i)  &  =q_{AB}(i)\\
\left[  p_{A}(i)-p_{B}(i+1)\right]  \widetilde{q}_{AB}^{\dag}(i)  &
=-\widetilde{q}_{AB}^{\dag}(i),
\end{align}
so the matrix entries $q_{AB}(i)$ and $\widetilde{q}_{AB}^{\dag}(i)$ cannot
simultaneously be non-zero. For generic values of the parameter $\varepsilon$,
we then see that both $\widehat{q}_{\varepsilon}(i)$ and $\widehat{q}%
_{-\varepsilon}(i)$ are invertible. Assuming this is the case, we can now
solve for the next value of $p(i+1)$ in terms of the previous time step:%
\begin{equation}
p(i+1)=\widehat{q}_{\varepsilon}^{-1}(i)p(i)\widehat{q}_{\varepsilon
}(i)-\widehat{q}_{\varepsilon}^{-1}(i)\widehat{q}_{-\varepsilon}(i).
\end{equation}
By the same token, we can also evaluate the new value of the \textquotedblleft
norm\textquotedblright\ $q(i+1)q^{\dag}(i+1)-\widetilde{q}^{\dag
}(i+1)\widetilde{q}(i+1)$:%
\begin{equation}
q(i+1)q^{\dag}(i+1)-\widetilde{q}^{\dag}(i+1)\widetilde{q}(i+1)=p(i+1)+q^{\dag
}(i)q(i)-\widetilde{q}(i)\widetilde{q}^{\dag}(i). \label{Drecurse}%
\end{equation}
To cast this in the form of a more explicit recursion relation, it is helpful
to work in terms of the polar decomposition for the matrices $q$ and
$\widetilde{q}$. For a complex matrix $q(i)$, this amounts to introducing a
unitary matrix $W(i)$ (complex phases) and a positive definite matrix $R(i)$ (the norm)
so that:
\begin{equation}
q(i)=R(i)W(i)\text{, \ \ }\widetilde{q}(i)=\widetilde{W}(i)\widetilde{R}(i),
\end{equation}
where we have also introduced a similar relation for $\widetilde{q}(i)$. In
terms of these matrices, the recursion relations read as:%
\begin{align}
p(i+1)  &  =\widehat{q}_{\varepsilon}^{-1}(i)p(i)\widehat{q}_{\varepsilon
}(i)-\widehat{q}_{\varepsilon}^{-1}(i)\widehat{q}_{-\varepsilon}(i).\\
R^{2}(i+1)-\widetilde{R}^{2}(i+1)  &  =p(i+1)+W^{\dag}(i)R^{2}%
(i)W(i)-\widetilde{W}(i)\widetilde{R}^{2}(i)\widetilde{W}^{\dag}(i),
\end{align}
So, provided we know $R$, $W$, $\widetilde{R}$ and $\widetilde{W}$ at a given
time step, we can feed this in and solve for the next time step. In practice,
of course, this is still a rather involved procedure, but in principle it
shows that a dynamical system continues to persist for all $N$.

Note that the recursion relation requires an additional input whenever we
cannot solve the system, i.e., when $\widehat{q}_{\varepsilon}$ does not
possess an inverse, and this is the higher rank analog of having a terminal puncture
in the $N = 1$ case.

It is also instructive to consider the evolution of the \textquotedblleft
center of mass\textquotedblright\ degree of freedom for the system. Taking the
trace and dividing by $N$, we obtain:%
\begin{align}
\frac{1}{N}\text{Tr}\left[  R^{2}(i+1)-\widetilde{R}^{2}(i+1)\right]   &
=\frac{1}{N}\text{Tr}\left[  p(i+1)\right]  +\frac{1}{N}\text{Tr}\left[
R^{2}(i)-\widetilde{R}^{2}(i)\right] \\
\frac{1}{N}\text{Tr}\left[  p(i+1)\right]   &  =\frac{1}{N}\text{Tr}\left[
p(i)\right]  -\frac{1}{N}\text{Tr}\left[  \widehat{q}_{\varepsilon}%
^{-1}(i)\widehat{q}_{-\varepsilon}(i)\right]  . \label{momentumevolve}%
\end{align}
This is not quite a closed system, because the second term on the righthand
side of line (\ref{momentumevolve}) is a number between $-1$ and $+1$. Said
differently, we need to include the higher order matrix powers to fully fix
this parameter. Note, however, that at large $N$, this last term can
approximate an arbitrary real number in the interval $[-1,+1]$, so in this
sense the jumps in the momenta can become quite small and uniform. At small
$N$, however, the system is highly discretized, and the jumps in momenta are
clearly more pronounced.

\section{Conclusions \label{sec:CONC}}

One of the important pieces of defining data for compactifications of 6D SCFTs on
Riemann surfaces involves the choice of boundary data at marked points.
In this work we have classified the structure of punctures in
perhaps the simplest case of a single M5-brane probing a $\mathbb{C}%
^{2}/\mathbb{Z}_{k}$ singularity. Starting from these solutions, we also
obtain a broad class of punctures for multiple M5-branes probing general ADE singularities.
The essential idea in our approach is to recast the structure of punctures in
terms of a dynamical system with position $x(i)$ and momentum $p(i)$ in which
evolution in time corresponds to moving from one node of an affine A-type
quiver to its neighbor. Classification of regular punctures (namely simple
poles) thus reduces to determining which initial conditions produce periodic
orbits of this dynamical system. Higher order poles for punctures follow
recursively from the lower order pole solutions. We have shown that for
periodic orbits, the momenta always take values in the rational numbers and
furthermore, satisfy the condition $p(i)=a(i)/b$ with $a(i),b$ relatively prime
integers and $b$ not divisible by $4$. Moreover, the resulting dynamical system
exhibits a remarkable sensitivity to initial conditions.
In the remainder of this section we discuss some future avenues for investigation.

One direction which would clearly be interesting to understand better would be
the structure of regular punctures for a single M5-brane probing a D- or
E-type singularity, and in particular the associated
dynamical system. At least for D-type theories, there is a similar notion of
\textquotedblleft long time evolution.\textquotedblright\ From this perspective,
the fact that the E-type theories have a small number of nodes suggests that
it should be possible to fully classify this case.

From a conceptual point of view, it is quite tempting to interpret the
time evolution of our dynamical system
directly in the geometry of the target space. Indeed, we can identify the
Taub-NUT circle of the resolution of
$\mathbb{C}^{2}/\mathbb{Z}_{k}$ with this time coordinate.
Along these lines, we have also seen that at least for sufficiently long orbits (namely for appropriate choices of the initial conditions), the analytic behavior of the dynamical system greatly simplifies. Geometrically, this appears to correspond to a limit where the orbifold $\mathbb{C}^{2}/\mathbb{Z}_{k}$ collapses to an $\mathbb{R}^3$. The T-dual description is also illuminating
as it takes us to a large number of parallel NS5-branes. It would be interesting to understand the arithmetic properties of the dynamical
system from this perspective.

In the large $N$ limit of many M5-branes, these probe theories have
holographic duals of the form $AdS_{7}\times S^{4}/\Gamma$ (see e.g. \cite{Gaiotto:2014lca,
DelZotto:2014hpa, Apruzzi:2013yva}), and
compactification on a Riemann surface with punctures have been studied in part
in \cite{Bah:2017wxp, Gaiotto:2009gz, Apruzzi:2015wna}. As we have already remarked, we can take our
dynamical system and embed it into higher rank theories. It would be very
interesting to understand the holographic description of this dynamical system.

Finally, we have seen that the structure of the punctures depends in a very sensitive way on the initial
momentum $p(1)$ of the dynamical system.  Indeed, the sort of branching behavior exhibited in figure \ref{fig:orbittree}
is reminiscent of the ultrametric trees which appear in some measures of complexity for spin glasses, and certain
enumeration problems connected with string vacua (see e.g. \cite{Denef:2011ee}). Here, we have found a
very concrete and well controlled example of this phenomenon in compactifications of 6D SCFTs.
The surprising complexity thus obtained would be exciting to further quantify.

\newpage

\section*{Acknowledgements}

We thank F. Apruzzi, J. Marincel, N. Miller and T. Rudelius for helpful
discussions. JJH thanks the 2017 Summer workshop at the Simons Center for
Geometry and Physics as well as the Aspen Center for Physics Winter Conference
in 2017 on Superconformal Field Theories in $d\geq4$, NSF grant PHY-1066293,
for hospitality during part of this work. The work of FH and JJH is supported
by NSF CAREER grant PHY-1756996. The work of FH is also supported by NSF grant PHY-1620311.

\appendix

\section{Periodic Orbit Proofs} \label{app:DOMINO}

With notation as in subsection \ref{porbitconds},
in this Appendix we establish that for an eternal puncture with initial momentum $p(1) = a / b$ with
$a,b$ coprime, when $b \neq 0$ mod $4$, we always obtain a periodic orbit.

We first establish some important properties of
the leading order contribution to the dynamical system, as given
in equation \ref{solDeltaL}.
\begin{lemma}\label{lemma1}
  Let $p=a/b$ where $a>0$ and $b>0$ are coprime. If $b\,\text{mod} \, 4%
  \ne 0$, there exists a periodic orbit for the leading contribution
  $\Delta^{(m)}_\mathrm{L}$ which is independent of the initial value
  $\Delta^{(1)}_\mathrm{L}$. The number of extrema in this orbit is given
  by
  \begin{equation}\label{kprime}
    k' = \begin{cases} b & b \text{ is even} \\
      2 b & b \text{ is odd}
    \end{cases}
  \end{equation}
  for which we have
  \begin{equation}\label{eqn:periodicDeltaL}
    \Delta^{(m+k')}_\mathrm{L} = \Delta^{(m)}_\mathrm{L}\,.
  \end{equation}
\end{lemma}

\begin{proof}[Proof of Lemma \ref{lemma1}]
In order to proof this lemma, we have to show that
\begin{equation}\label{involvedsum}
   2 k' p + \sum\limits_{l=0}^{k'-1} (-1)^l \,
    \text{Ceil}(\Delta^{(m)}_\mathrm{L} + 2 l p) = 0
\end{equation}
holds under the assumptions stated in the lemma. If this is the case,
\eqref{eqn:periodicDeltaL} follows directly from \eqref{solDeltaL}. Most of
the terms in the sum cancel due to the alternating sign. If we only keep the
non-vanishing contributions, the equation we have to check can be rewritten as
\begin{equation}\label{simplersum}
  2 k' p + \sum\limits_{l=0}^{k'-1} (-1)^l \, \text{Ceil}(\Delta^{(m)}_\mathrm{L}
    + 2 l p) =   \sum\limits_{l=0}^{2 k' p - 1} \sgn_- \left(
    2 p - ( l \, \mathrm{mod}\, 4 p ) \right) = 0
\end{equation}
with
\begin{equation}
  \sgn_- (x) = \begin{cases} \phantom{-} 1 & x>0 \\
    -1 & x \le 0\,. \end{cases}
\end{equation}
The argument of the $\sgn_-$ function has the point symmetry
\begin{equation}
  2 p - ( ( l_* + l ) \, \mathrm{mod} \, 4 p ) = - 2 p + ( ( l_* - l ) \,
    \mathrm{mod} \, 4 p ) \quad \forall \,
    l \in \mathbb{N}\,, \, 0 < l < 2 k' p - l_*
\end{equation}
around the point
\begin{equation}
  l_* = \begin{cases}
    a & \text{for } b \, \mathrm{mod} \, 2 = 0 \\
    2 a & \text{for } b \, \mathrm{mod} \, 2 \ne 0\,.
  \end{cases}
\end{equation}
Hence, there are even more terms in the sum \eqref{simplersum} which cancel, too.
Finally, one is left with
\begin{align}
  \sum\limits_{l=0}^{2 k' p - 1} \sgn_- \left( 2 p - ( l \, \mathrm{mod}\, 4 p )
    \right) &= \sum\limits_{l\in\{0,\,l_*\}} \sgn_- \left( 2 p - ( l \,
    \mathrm{mod}\, 4 p ) \right) \nonumber \\
  &= \sgn_- ( 2 p ) - \sgn_- (0) = 1 - 1 = 0\,.
\end{align}
\end{proof}

\begin{lemma}\label{lemma2}
  Let $p=a/b$ where $a>0$ and $b>0$ are coprime. If $b\,\text{mod}\, 4 = 0$,
  $\Delta^{(m)}_\mathrm{L}$ is shifted by $1$ after $b/2$ time steps
  \begin{equation}
    \Delta^{(m+b/2)}_\mathrm{L} = \Delta^{(m)}_\mathrm{L} + 1\,.
  \end{equation}
\end{lemma}

\begin{proof}[Proof of Lemma \ref{lemma2}]
The proof for this lemma, requires only a very slight modification of the
previous proof. Instead of \eqref{involvedsum}, we now want to show that
\begin{equation}
  b p + \sum\limits_{l=0}^{b/2-1} (-1)^l \, \text{Ceil}(\Delta^{(m)}_\mathrm{L}
    + 2 l p) = 1
\end{equation}
holds. Again, we rewrite this sum as
\begin{equation}
  b p + \sum\limits_{l=0}^{b/2-1} (-1)^l \, \text{Ceil}(\Delta^{(m)}_\mathrm{L} +
    2 l p) =  \sum\limits_{l=0}^{b p - 1} \sgn_- \left( 2 p -
    ( l \, \mathrm{mod}\, 4 p ) \right) = 1
\end{equation}
and take advantage of the point symmetry of the argument of the $\sgn_-$ function.
Now, $l_*=a/2$ and therewith not an integer. Thus, instead of two contributions to the
same which cancel each other, we find
\begin{equation}
  b p + \sum\limits_{l=0}^{b/2-1} (-1)^l \, \text{Ceil}(\Delta^{(m)}_\mathrm{L} +
    2 l p) = \sgn_- (2p) = 1\,.
\end{equation}
\end{proof}

\noindent Now consider a periodic orbit $\{\Delta^{(1)}_\mathrm{L}, \dots,%
\Delta^{(k')}_\mathrm{L}\}$ which arises from lemma \ref{lemma1} and take
instead of $\Delta^{(1)}_\mathrm{L}$, $\Delta^{(2)}_\mathrm{L}$ as the initial
condition. In this case, we obtain another periodic orbit $\{\Delta^{(2)}_\mathrm{L},%
\dots, \Delta^{(k')}_\mathrm{L}, \Delta^{(1)}_\mathrm{L}\}$. The only
difference between the two of them is that all elements are shifted to the
right by one position. Periodic orbits which arise from such shifts are members
of the same equivalence class. All their relevant properties do not change within
the class. Hence, it is sufficient to study only one representative of each
equivalence class.
\begin{lemma}\label{lemma3}
  A unique representative of each equivalence class for periodic orbits
  $\Delta^{(m)}_\mathrm{L}$ which arise from lemma \ref{lemma1} is given by
  the initial conditions
  \begin{equation}\label{classrepresentative}
    0 < |\Delta^{(1)}_\mathrm{L}| < \frac1{k'}
  \end{equation}
  for $p=a/b$, $a>0$ and $b>0$. For all other remaining $\Delta^{(m)}_\mathrm{L}$
  in this orbit we instead have:
  \begin{equation}\label{goodpoints}
    \left| \left(\Delta^{(m)}_\mathrm{L} + \frac12\right) \,\text{mod}\, 1 -%
      \frac12 \right| > \frac1{k'}
      \quad \forall \, m \in \{2, \dots,k'/2,k'/2+2, \dots,  k'\}
  \end{equation}
  and
  \begin{equation}\label{badpoint}
    \left| \left(\Delta^{(k'/2+1)}_\mathrm{L} + \frac12\right) \,\text{mod}\, 1 -%
      \frac12 \right| < \frac1{k'}\,.
  \end{equation}
\end{lemma}

\begin{proof}[Proof of Lemma \ref{lemma3}]
First we prove that from the representatives \eqref{classrepresentative} all
other orbits of the same equivalence class can be obtained by shifting the time
evolution. More specifically, we shift by $n$ steps to the left
\begin{equation}
  \Delta'^{(m)}_\mathrm{L} = \Delta^{(m+2n)}_\mathrm{L} - 2 \, \text{Floor}\left(
    \frac{\Delta^{(1+2 n)}_\mathrm{L} + 1}2 \right)
\end{equation}
and perform an additional integer shift to keep $\Delta'^{(1)}$ in the fixed
interval $(-1,1)$. According to \eqref{boundDeltaL1}, this interval covers all
possible initial conditions. A shift by an integer leaves the iteration
\eqref{leading} for the leading contributions $\Delta^{(m)}_\mathrm{L}$ invariant
and therefore does not change the equivalence class of the orbit. After this
transformation, we obtain
\begin{equation}\label{shiftedinitialcond}
  \Delta'^{(1)}_\mathrm{L} = ( \Delta^{(1)}_\mathrm{L} + 2 n p + 1)
    \, \text{mod} \, 2 - 1\,.
\end{equation}
as new initial condition. Sweeping $n$ from $1$ to $k'$, we find that all
initial conditions fill the lattice
\begin{equation}
  \Delta'^{(1)}_\mathrm{L} \in \{  -1 + 1 / k' + \Delta^{(1)}_\mathrm{L}, \,
    -1/2 + 3 / k' + \Delta^{(1)}_\mathrm{L}, \dots, 1 - 1 / k' +
    \Delta^{(1)}_\mathrm{L} \}\,.
\end{equation}
Varying also $\Delta^{(1)}_\mathrm{L}$ in the bounds given by
\eqref{classrepresentative}, we eventually see that
\begin{equation}
  \forall \, \Delta'^{(1)}_\mathrm{L} \in ( -1 , 1)
    \quad \exists \, n \in \{1, \dots , k'\} \text{ and }
    \Delta^{(1)}_\mathrm{L} \in \left(-\frac1{k'}, \frac1{k'}\right)
    \text{ s.t. \eqref{shiftedinitialcond} holds.}
\end{equation}
This proves the first part of the lemma.

In order to prove the second part, we first note that
\begin{equation}
  \left(\Delta^{(m)}_\mathrm{L} + \frac12\right) \,\text{mod}\, 1 =
  \left(\Delta^{(1)}_\mathrm{L} + 2 (m - 1) p + \frac12\right) \,\text{mod}\, 1
    \quad \forall \, m \in \mathbb{Z}
\end{equation}
which follows from \eqref{leading} after stripping off all integer
contributions. Because
\begin{equation}
  k' p \in \mathbb{N}\,,
\end{equation}
we immediately find that
\begin{equation}\label{otherhalf}
  \left(\Delta^{(m + k'/2)}_\mathrm{L} + \frac12\right) \,\text{mod}\, 1 =
    \left(\Delta^{(m)}_\mathrm{L} + \frac12\right) \,\text{mod}\, 1
\end{equation}
holds. Under the assumption \eqref{classrepresentative}, this relation directly
implies \eqref{badpoint}. More generally,
\begin{equation}\label{latticevalues}
  \left(\Delta^{(1)}_\mathrm{L} + 2 (m - 1) p + \frac12\right) \,\text{mod}\, 1
    - \frac12 \in \Lambda_{k'}(\Delta^{(1)}_\mathrm{L})
\end{equation}
only takes values on the lattice
\begin{equation}
  \Lambda_{k'}(\Delta) = \{ -1/2 + 1 / k' + \Delta,\, -1/2 + 3 / k' + \Delta,
    \dots, 1/2 - 1 / k' + \Delta \}\,.
\end{equation}
with $|\Delta|<1/k'$. This lattice has exactly $k'/2$ elements and for
$m \in \{1, \dots,k'/2 \}$ every one occurs exactly one times. There is only
one element in this lattice whose absolute value is small that $1/k'$. This
element is $\Delta$. It arise for $m=0$ from \eqref{latticevalues}. All other
elements have an absolute value which is larger that $1/k'$. This proves
\eqref{goodpoints} for $m\in\{2, \dots, k'/2\}$. The remaining conditions follow
immediately from \eqref{otherhalf}.
\end{proof}

\noindent We call this canonical representative an aligned periodic orbit:
\begin{definition}\label{aligned}
  A periodic orbit ($p=a/b$, $a$, $b$ coprime and $b\, \text{mod} \, 4 \ne 0$) with
  \begin{equation}
    0 < |\Delta^{(1)}_\mathrm{L}| < \frac1{k'}
  \end{equation}
  is called aligned.
\end{definition}
\noindent A direct consequence of lemma \ref{lemma3} is that every orbit can be
aligned by a suitable shift.
\begin{lemma}\label{lemma4}
  Let $\Delta^{(m)}_\mathrm{L}$ capture a periodic orbit which arises from lemma
  \ref{lemma1}. All $\Delta^{(m)}_\mathrm{L}$ are bounded by
  \begin{equation}
    | \Delta^{(m)}_\mathrm{L} | < k' p + 2 = \begin{cases} a + 1 &
        \text{for } b \text{ even} \\
      2 a + 1 & \text{for } b \text{ odd}\,.
    \end{cases}
  \end{equation}
\end{lemma}

\begin{proof}[Proof of Lemma \ref{lemma4}]
 First, we deduce for \eqref{leading} that
\begin{equation}
  | \Delta^{(m+1)} | < |\Delta^{(m)} | + 1
\end{equation}
holds. In order to prove the lemma, we now show that
\begin{equation}\label{wantedlemma4}
  | \Delta^{(1+2n)} | < k' p + 1 \quad \forall \, n \in \mathbb{Z}\,.
\end{equation}
As in the proofs of the lemmas \ref{lemma1} and \ref{lemma2}, we write the
expression for $\Delta^{(1+2n)}_\mathrm{L}$ in \eqref{solDeltaL} as
\begin{equation}\label{deltaL1+2n}
	\Delta^{(1+2 n)}_\mathrm{L} =( \Delta^{(1)}_\mathrm{L} + 4 n p ) \,
		\text{mod} \, 1 + \sum\limits_{l=0}^{\text{Floor}( 4 n p )  - 1}
		\sgn_- \, ( 2 p - ( l \, \text{mod} \, 4p ) )\,.
\end{equation}
For $n=k'/2$ the sum on the right hand side vanishes according to lemma
\ref{lemma1}. Moreover the $\sgn_-$ function can only take values of $1$
or $-1$. Hence, the sum is bounded by
\begin{equation}
	\left| \sum\limits_{l=0}^{\text{Floor}( 4 n p )  - 1} \sgn_- \, ( 2 p
		- ( l \, \text{mod} \, 4p ) ) \right| < k' p \quad \forall n \in \mathbb{Z}
		\,.
\end{equation}
Taking into account the first part of \eqref{deltaL1+2n} too, we see that
\eqref{wantedlemma4} hold. This completes the proof.
\end{proof}

Let us now start to analyze the effect of the corrections $\delta^{(m)}$ in
\eqref{DeltaSplitting} for the case where the leading contributions
$\Delta^{(m)}_\mathrm{L}$ describe an aligned, periodic orbit defined in
definition \ref{aligned}. In this discussion the functions $\delta_1(\Delta, y)$,
$\delta_2(\Delta, y)$ and $\tilde \delta_2(\Delta, y)$, which we defined in
\eqref{delta1}, \eqref{delta2} and \eqref{tildedelta2}, play an important role.
Under the assumption that $\Delta$ is bound by
\begin{equation}
  |\Delta| < \Delta_\mathrm{B} \label{DeltaBOUND}\,,
\end{equation}
we can bound them from above by
\begin{equation}\label{BOUNDS}
  |\delta_1(\Delta, y)| < \frac{(2 \Delta_\mathrm{B} + 2)^2}{2 (l - 1)}\,, \quad
  |\delta_2(\Delta, y)| < \frac{(2 \Delta_\mathrm{B} + 2)^2}{2 (l - 1)}
    \quad \text{and} \quad
  |\tilde\delta_2(\Delta, y)| < \frac{(2 \Delta_\mathrm{B} + 2)^2}{2 (l - 1)} \,.
\end{equation}
We want to guarantee that the subleading part $|\delta^{(m)}|$ is small and
does not change the qualitative behavior of the time evolution. If this is the
case, the time evolution is governed by the leading contributions
$\Delta^{(m)}_\mathrm{L}$ for which we already proved an analytic expression
in \eqref{solDeltaL} above. More specifically, $|\delta^{(m)}|$ has to be smaller
than a certain boundary $\delta_\mathrm{B}<1$, at least for the first $k'/2$
steps. Thus, we require
\begin{equation}\label{boundsoncorrections}
  |\delta^{(m)}| < \delta_\mathrm{B} \quad \text{and} \quad
  |\Delta^{(m)}| < \Delta_\mathrm{B}  \quad \forall \,
    m \in \{ 1, \dots, k'/2 + 1\} \,.
\end{equation}
and fix $\delta_\mathrm{B}$ such that
\begin{equation}\label{constrceil}
  \text{Ceil}( \Delta^{(m)}_\mathrm{L} ) = \text{Ceil}\big( \Delta^{(m)} +%
    \delta_1( \Delta^{(m)}, y^{(m)} ) \big)
    \quad \forall \, m \in \{ 1, \dots, k'/2 \}
\end{equation}
is not violated. If it would be violated than $|\delta^{(m)}|$ would be
immediately bigger than one and our analysis would break down.
For $m=1$, we have
\begin{equation}
  \text{Ceil}( \Delta^{(1)}_\mathrm{L} ) = \text{Ceil}\big( \Delta^{(1)}_\mathrm{L}
    - \delta_1( \Delta^{(1)}, y^{(1)} )  + \delta_1( \Delta^{(1)}, y^{(1)} ) \big)
    = \text{Ceil}( \Delta^{(1)}_\mathrm{L} )
\end{equation}
automatically. For $m=2$, we obtain the constraint
\begin{equation}
  \text{Ceil}( \Delta^{(2)}_\mathrm{L} ) = \text{Ceil}\big( \Delta^{(2)}_\mathrm{L}
    + \delta^{(2)} + \delta_1( \Delta^{(2)}, y^{(2)} ) \big)
    = \text{Ceil}( \Delta^{(2)}_\mathrm{L} )
\end{equation}
where
\begin{equation}
  |\delta^{(2)}| \le | \delta^{(1)} + \delta_2(\Delta^{(1)}, y^{(1)}) |
    < 2 \frac{(2\Delta_\mathrm{B} + 2)^2}{2(l-1)}
\end{equation}
after taking into account the bounds \eqref{BOUNDS}. This constraint only
holds, if we restrict the domain of allowed values for $\Delta^{(2)}_\mathrm{L}$
to
\begin{equation}
  \left| \left(\Delta^{(2)}_\mathrm{L} + \frac12\right) \,\text{mod}\, 1 -
    \frac12 \right| \ge 3 \frac{(2 \Delta_\mathrm{B} + 2)^2}{2(l-1)}\,.
\end{equation}
Repeating this reasoning step for step, we eventually find
\begin{equation}
  \bigwedge\limits_{m=2}^{k'/2}  \left| \left(\Delta^{(m)}_\mathrm{L} +
    \frac12\right) \,\text{mod}\, 1  - \frac12  \right|
    \ge (m+1) \frac{(2 \Delta_\mathrm{B} + 2)^2}{2(l-1)}\,,
\end{equation}
which implies
\begin{equation}
  | \delta^{(k'/2+1)} | < k' \frac{(\Delta_\mathrm{B} + 1)^2}{l-1}
    = \delta_\mathrm{B}\,.
\end{equation}

For the following discussion it is sufficient to use the weaker, but
simpler version
\begin{equation}\label{restrDelta}
  \left| \left(\Delta^{(m)}_\mathrm{L} + \frac12\right) \,\text{mod}\, 1 -
    \frac12 \right| \ge (k' + 2) \frac{( \Delta_\mathrm{B} + 1)^2}{l-1}
    = \delta_\mathrm{B} \quad \forall \, m \in \{ 2, \dots, k'/2 \}\,.
\end{equation}
If it holds, the subleading part $\delta^{(m)}$ is given by
\begin{equation}\label{subleading}
  \delta^{(m)} = \delta^{(1)} + \sum\limits_{n=1}^{m-1} \delta_2\big(
    \Delta_\mathrm{L}^{(n)}, y^{(n)} \big)
    \quad \text{for} \quad \forall \, m \in \{ 2, \dots, k'/2 + 1\}\,.
\end{equation}
Taking into account lemma \ref{lemma3}, we know that:
\begin{equation} \label{kprimebound}
  \left| \left(\Delta^{(m)}_\mathrm{L} - \frac12\right)
    \,\text{mod}\, 1 + \frac12 \right| > \frac{1}{k'}.
\end{equation}
On the other hand, the constraint \eqref{restrDelta} also tells us that:
\begin{equation}
  \delta_\mathrm{B} \le \left| \left(\Delta^{(m)}_\mathrm{L} - \frac12\right)
    \,\text{mod}\, 1 + \frac12 \right|
\end{equation}
Scanning over the admissible values of $\delta_{B}$, we see that the largest $\delta_{B}$ which will satisfy
line \ref{kprimebound} satisfies the inequality:
\begin{equation}
  \delta_\mathrm{B} \le \frac{1}{k'}\,.
\end{equation}
If we furthermore combine lemma \ref{lemma4} with $|\delta^{(m)}| <
\delta_\mathrm{B} < 1$ for $m\in\{1,\dots,k'/2\}$, we see that there are two additional
inequalities we can write:
\begin{align}
|\Delta^{(m)}| & < \Delta_\mathrm{B} \\
|\Delta^{(m)}| & \le |\Delta^{(m)}_\mathrm{L}| + |\delta^{(m)}| < k' p +  3 \quad \forall \, m \in \{ 2, \dots, k'/2 \}
\end{align}
holds. Again, we seek out the largest value of $\Delta_{B}$ compatible with these two conditions. This is
fixed by taking:
\begin{equation}\label{DeltaB}
  \Delta_\mathrm{B} = k' p + 3\,.
\end{equation}
Hence, we can express both boundaries $\delta_\mathrm{B}$ and
$\Delta_\mathrm{B}$ in \eqref{boundsoncorrections} in terms of
$k'$ and $p$.

So far we just discussed the subleading part $\delta^{(m)}$ for the first
half of the periodic orbit. To go beyond $\Delta^{(k'/2 + 1)}$ is more
involved, because according to lemma \ref{lemma3}
\begin{equation}
  \left| \left(\Delta^{(k'/2+1)}_\mathrm{L} + \frac12\right) \,\text{mod}\, 1
    - \frac12 \right| < \frac1{k'}
\end{equation}
holds and enables us to apply the bounds discussed above. There are now two
ways one could go. First one could improve the analysis of the contributions
from $\delta_1(\Delta,y)$ and $\delta_2(\Delta,y)$. We only used very crude
bounds. They worked well for the analysis of the first $k'/2$ corrections
$\delta^{(m)}$ but are insufficient as soon as one want to go beyond the
special point $k'/2 + 1$. Alternatively, we can avoid this point completely
by approaching the second part of the orbit with the inverse iteration
prescription \eqref{inviter}. In analogy with \eqref{boundsoncorrections},
one now requires
\begin{equation}
  | \delta^{(m)} | < \delta_B \quad \text{and} \quad | \Delta^{(m)} | <
    \Delta_B \quad \forall \, m \in \{- k'/2 + 1 , \dots , 0 \}\,.
\end{equation}
As before, we again require
\begin{equation}
  \text{Ceil}( \Delta^{(m)}_\mathrm{L} ) = \text{Ceil}\big( \Delta^{(m)} +
    \delta_1( \Delta^{(m)}, y^{(m)} ) \big)
    \quad \forall \, m \in \{ - k'/2 + 1, \dots, 0 \}
\end{equation}
in order to keep $\delta^{(m)}$ smaller than one. Following the same reasoning
as above, we find that this constraint is fulfilled when:
\begin{equation}
  \delta_\mathrm{B} = (k' + 2) \frac{(\Delta_\mathrm{B} + 1)^2}{l - 1} \le \frac1{k'}
\end{equation}
with the $\Delta_\mathrm{B}$ in \eqref{DeltaB}. Alternatively, we can also state this
inequation for $l$, resulting in
\begin{equation}\label{lbound}
  1 + k' ( k' + 2 ) ( k' p + 4 )^2 \le l\,.
\end{equation}
If it is satisfied, the corrections for the second part of the orbits are
given by
\begin{equation}
  \delta^{(m)} = \delta^{(1)} + \sum\limits_{n=m+1}^{1} \tilde\delta_2\big(
    \Delta_\mathrm{L}^{(n)}, y^{(n)} \big)
    \quad \text{for} \quad \forall \, m \in \{ - k'/2 + 1, \dots, 0\}\,.
\end{equation}
We call the corresponding periodic orbits dominant orbits.
\begin{definition}\label{dominantorbit}
  Orbits with $p=a/b$ where $a$, $b$ are coprime and $b \, \text{mod}\, 4 \ne 0$
  are called dominant if they fulfill
  \begin{equation}
    l \ge 9 + 2 k' p + k' ( k' + 2 )( k' p + 4 )^2 = \begin{cases}
      9 + 2 a + b ( b + 2 ) ( a + 4 )^2 & \text{for } b \text{ even} \\
      9 + 4 a + 4 b ( b + 1 ) ( 2 a + 4 )^2 & \text{for } b \text{ odd}\,.
    \end{cases}
  \end{equation}
\end{definition}
\noindent Note that this definition also takes orbits into account which are
not aligned. To this end, we require that not only $l^{(1)}=l$ satisfies the
bound \eqref{lbound} but all $l^{(m)}$. For dominant orbits we can now establish
\begin{lemma}
  For dominant orbits, the subleading part is given by
  \begin{equation}
    \delta^{(m)} = \delta^{(1)} + \sum\limits_{n=1}^{m-1} \delta_2\big(
      \Delta_\mathrm{L}^{(n)}, y^{(n)} \big)
      \quad \text{for} \quad \forall \, m \in \{ 2, \dots, k'/2 + 1\}\,.
  \end{equation}
  and
  \begin{equation}
    \delta^{(m)} = \delta^{(1)} + \sum\limits_{n=m+1}^{1} \tilde\delta_2\big(
      \Delta_\mathrm{L}^{(n)}, y^{(n)} \big)
      \quad \text{for} \quad \forall \, m \in \{ - k'/2 + 1, \dots, 0\}\,.
  \end{equation}
\end{lemma}

After all this preparation, we can finally state the important theorems for
dominant periodic orbits.
\begin{theorem}
  Dominant orbits are periodic and contain $k'$ (given by \eqref{kprime}) extrema.
\end{theorem}
\begin{proof}
  In order to prove that we indeed have a periodic orbit, we have to show that
  \begin{equation}
    \Delta^{(-k'/2 + 1)} = \Delta^{(k'/2 + 1)}
  \end{equation}
  holds. This relation is equivalent to
  \begin{equation}
    \Delta^{(-k'/2 + 1)}_\mathrm{L} + \delta^{(-k'/2 + 1)} =
      \Delta^{(k'/2 + 1)}_\mathrm{L} + \delta^{(k'/2 + 1)}\,.
  \end{equation}
  For lemma \ref{lemma1}, we know that $\Delta^{(-k'/2 + 1)}_\mathrm{L} =%
  \Delta^{(k'/2 + 1)}_\mathrm{L}$ holds. Therefore, we only have to show
  $\delta^{(-k'/2 + 1)} = \delta^{(k'/2 + 1)}$ which is equivalent to
  \begin{equation}\label{toshow}
    \sum\limits_{n = -k'/2 + 2}^1 \tilde\delta_2 (\Delta_\mathrm{L}^{(n)}, y^{(n)}) =
    \sum\limits_{n = 1}^{k'/2} \delta_2 (\Delta_\mathrm{L}^{(n)}, y^{(n)})\,.
  \end{equation}
  First, we note that lemma~\ref{lemma1} implies
  \begin{equation}
    \sum\limits_{n = -k'/2 + 2}^1 \tilde\delta_2 (\Delta_\mathrm{L}^{(n)}, y^{(n)}) =
    \sum\limits_{n = k'/2 + 2}^{k'+1} \tilde\delta_2 (\Delta_\mathrm{L}^{(n)}, y^{(n)})
  \end{equation}
  and furthermore
  \begin{equation}
    \tilde\delta_2 (\Delta_\mathrm{L}^{(n)}, y^{(n)}) =
      - \delta_2 (\Delta_\mathrm{L}^{(n-1)}, y^{(n-1)})
  \end{equation}
  follows form the definitions \eqref{delta2}, \eqref{tildedelta2} and the
  iteration prescription \eqref{leading} for the leading contribution.
  Therefore \eqref{toshow} is equivalent to
  \begin{align}
    0 &=\sum\limits_{n=1}^{k'} \delta_2(\Delta_\mathrm{L}^{(n)}, y^{(n)}) \\
    & = \frac{2}{l-1} \sum\limits_{n=0}^{k'/2 - 1} \left[ 3 \text{Ceil}(\Delta_\mathrm{L}^{(2n+1)})
      - \text{Ceil}(\Delta_\mathrm{L}^{(2n+1)}+2p) - 1 \right] \times \nonumber \\
    & \quad\quad\quad\quad\quad\quad\quad\quad\quad \left[ 2 p
      + \text{Ceil}(\Delta_\mathrm{L}^{(2n+1)}) - \text{Ceil}(\Delta_\mathrm{L}^{(2n+1)}+2p) \right]\,.
      \nonumber
  \end{align}
  After introducing
  \begin{equation}
    s^{(n)} = \text{Ceil}(\Delta_\mathrm{L}^{(1)} + 4 p n) - \text{Ceil}(\Delta_\mathrm{L}^{(1)}
      + 4 p n + 2 p )\,,
  \end{equation}
  this equation is equivalent to
  \begin{align}
    0 = \sum\limits_{n=0}^{k'/2 - 1} 2 \text{Ceil}(\Delta_\mathrm{L}^{(2n+1)}) (2 p + s^{(n)}) +
        \sum\limits_{n=0}^{k'/2 - 1} (2 p + s^{(n)})(s^{(n)} - 1)\,.
  \end{align}
  Writing \eqref{solDeltaL} as
  \begin{equation}
    \Delta_\mathrm{L}^{(1+2n)} = \Delta_\mathrm{L}^{(1)} + 2 \sum\limits_{l=0}^{n-1} s^{(l)}\,,
  \end{equation}
  lemma \ref{lemma1} implies
  \begin{equation}
    \sum\limits_{l=0}^{k'/2+1} s^{(l)} = - k' p\,.
  \end{equation}
  Moreover, we take into account that $(s^{(l)})^2 = - s^{(l)}$. Thus, \eqref{toshow}
  is equivalent to
  \begin{align}
    \sum\limits_{n=0}^{k'/2-1} \text{Ceil}(\Delta_\mathrm{L}^{(2n+1)})(2 p + s^{(n)}) &=
    \sum\limits_{n=0}^{k'/2-1} \left(\text{Ceil}(\Delta_\mathrm{L}^{(1)} + 4 p n) + 2 \sum\limits_l^{n-1}
      s^{(l)} \right) \left( 2 p + s^{(n)} \right) \nonumber \\
    &= \frac{k' p}2 ( 2 p - 1)\,.
  \end{align}
  All what remains is to calculate the four remaining sums on the right hand side in the first line.
  Let us start with
  \begin{equation}\label{sum1}
    2 p \sum\limits_{n=0}^{k'/2-1} \text{Ceil}(\Delta_\mathrm{L}^{(1)}+4pn) = \frac{k'p}2 (2 k'p + 1)
      -2 k' p^2 + p \, \text{sgn }\Delta_\mathrm{L}^{(1)}
      \quad \text{for} \quad 0 < |\Delta_\mathrm{L}^{(1)}| < \frac1{k'}\,.
  \end{equation}
  In order to better understand this result, assume for a moment that $\Delta_\mathrm{L}^{(1)}=0$.
  Now
  \begin{equation}
    \text{Ceil}(4 n p) + \text{Ceil}\left( 4 ( k'/2 - n) p \right) = 2 k' p + 1
  \end{equation}
  holds. We find this contribution exactly $k'/4-1/2$ times in the sum in \eqref{sum1}. If
  $|\Delta_\mathrm{L}^{(1)}|<1/k'$ instead of being zero, it only affects the summand for $n=0$
  and therefore gives rise to the $\text{sgn}$ function in \eqref{sum1}. Second, we evaluate the
  sum
  \begin{equation}
    2 \sum\limits_{n=0}^{k'/2-1} \sum\limits_{l=0}^{n-1} s^{(n)} s^{(l)} =
    2 \sum\limits_{n=1}^{k' p} (n-1) = k'p(k'p - 1)\,.
  \end{equation}
  Thus, we finally need to show that
  \begin{equation}\label{sum34}
    \sum\limits_{n=0}^{k'/2-1} \left( \text{Ceil}(\Delta_\mathrm{L}^{(1)}+4pn) s^{(n)} +
    4 p \sum\limits_{l=0}^{n-1} s^{(l)}\right) =
    - k' p ( 2 k' p - 3 p) - p \text{sgn }(\Delta_\mathrm{L}^{(1)})
  \end{equation}
  holds. To do, let us consider the following problem: Choose $i=1, \dots, k'p$
  integers $0 \le m_i \le 2 k' p$ such that they fulfill
  \begin{equation}\label{condmi}
    \exists\, n_i \in \mathbb{N} \text{ s.t. }
    \Delta_\mathrm{L}^{(1)} + 4 p n_i < m_i < \Delta_\mathrm{L}^{(1)} + 4 p n_i + 2 p\,.
  \end{equation}
  In terms of these integers, we can write \eqref{sum34} as
  \begin{equation}
    \sum\limits_{i=1}^{k'p} \left( 4 p\,\text{Ceil} \left(
      \frac{m_i - \Delta_\mathrm{L}^{(1)}}{4p}\right) - m_i - 2 k' p \right) =
      - k' p ( 2 k' p - 3 p) - p \, \text{sgn }(\Delta_\mathrm{L}^{(1)})
  \end{equation}
  which simplifies to
  \begin{equation}\label{finalsum}
    4 \sum\limits_{i=1}^{k'p} \left( \frac{m_i - \Delta_\mathrm{L}^{(1)}}{4 p}
      \text{ mod } 1 \right) = k'p + \text{sgn } \Delta_\mathrm{L}^{(1)} - \Delta_\mathrm{L}^{(1)} k'\,.
  \end{equation}
  To see that this relations is indeed fulfilled, we again check first the case for
  $\Delta_\mathrm{L}^{(1)} = 0$. In this case, it is is easy to check that if an integer $m$
  satisfies \eqref{condmi} there is another integer
  \begin{equation}
    m' = ( k' p - m ) \text{ mod } 2 k' p
  \end{equation}
  which does so, too. Let us add up this two contributions to the sum on the left-hand
  side of \eqref{finalsum}:
  \begin{equation}
    \frac{m}{4 p} \text{ mod } 1 +
    \left( \frac{k'}4 - \frac{m}{4 p} \right) \text{ mod } 1 = \frac12\,.
  \end{equation}
  Here, we have used that
  \begin{equation}
    \frac{k'}4 \text{ mod } 1 = \frac12 \quad \text{and} \quad
    0 \le \frac{m}{4p} \text{ mod } < \frac12\,,
  \end{equation}
  which follows directly from \eqref{condmi}. This situation occurs $k' p/2$ times and we
  reproduces \eqref{finalsum} for $\Delta_\mathrm{L}^{(1)}=0$. For $|\Delta_\mathrm{L}^{(1)}|<\frac1{k'}$,
  the deviations from this argumentation are minor. It is straightforward to check that
  they reproduce exactly the remaining terms on the right-hand side of \eqref{finalsum}.
\end{proof}
\begin{corollary}
  All orbits with $p = a / b$ with $a$, $b$ coprime and
  $b \, \text{mod} \, 4 \ne 0$ are periodic.
\end{corollary}
\begin{proof}
We prove this statement by contradiction. Assume we start from a $y^{(1)}$
which does not give rise to a periodic orbit. Of course this $y^{(1)}$ can
not be the initial condition of a dominant orbit because these are periodic.
Thus
\begin{equation}
  |y^{(1)}| < y_\mathrm{max}
\end{equation}
is bounded from above (with increasing $|y^{(1)}|$ also $l$ increases gradually
and at some point would fulfill the requirement in definition
\ref{dominantorbit} for a dominant orbit). Moreover, we note that $y^{(m)}$ can
only change in multiples of $1/b$ according to \eqref{iteration}. As the
sequence $y^{(m)}$ does not describes a periodic orbit, all $y^{(m)}$ are unique
and can not repeat. These two observations give rise to
\begin{equation}
  \exists \, y \in \{ y^{(2)}, \dots , y^{(N)} \} \quad \text{with} \quad
    | y - y^{(1)} | \ge \frac{N}{2 b}
\end{equation}
which implies
\begin{equation}
  | y | \ge \frac{N}{2 b} - y_\mathrm{max}\,.
\end{equation}
For $N=\text{Ceil}( 4 b y_\mathrm{max})$ there exists an $y$ in the first $N$
elements of the time evolution which is bigger that $y_\mathrm{max}$. This is
a contradiction. Such a $y$ implies a periodic, dominant orbit, but we started
with the assumption that the orbit is not periodic.
\end{proof}
\begin{theorem}
  Let $x^{(i)}$ describe a dominant, aligned, periodic orbit. Its length, which is defined as
  \begin{equation}
    x(i) = x(i+k)\,,
  \end{equation}
  is
  \begin{equation}
    k = k' ( l + 2 p ) = \begin{cases} l b + 2 a & \text{for } b \text{ even} \\
      2 l b + 4 a & \text{for } b \text{ odd}\,.
    \end{cases}
  \end{equation}
\end{theorem}
\begin{proof}
  The length of the orbit is given by
  \begin{equation}
    k = \sum\limits_{l=1}^{k'} l^{(m)} = \sum\limits_{l=0}^{k'/2-1} \left[ l^{(1+2 l)} + l^{(2 + 2l)} \right]\,.
  \end{equation}
  It is straightforward to check that
  \begin{equation}
    l^{(m)} + l^{(m+1)} = 2 l - 2 \text{Ceil} ( \Delta^{(m)}_\mathrm{L} ) + 2 \text{Ceil} ( \Delta^{(m)}_\mathrm{L} + 2p )
  \end{equation}
  follows from \eqref{l(m)new}. It allows us to write
  \begin{equation}
    k = l k' - \sum\limits_{l=0}^{k'-1} (-1)^l \, \text{Ceil} ( \Delta^{(1)}_\mathrm{L} + 2 l p )
      = l k' + \Delta^{(1)}_\mathrm{L} - \Delta_{\mathrm{L}}^{(1+k')} + 2 k' p
  \end{equation}
  after taking into account \eqref{solDeltaL}. According to lemma \ref{lemma1} the part $\Delta^{(1)}_\mathrm{L} - \Delta_{\mathrm{L}}^{(1+k')}$ vanishes.
\end{proof}

\newpage

\bibliographystyle{utphys}
\bibliography{DYNAMO}

\end{document}